\documentclass[10pt,journal,compsoc]{IEEEtran}
\ifCLASSOPTIONcompsoc
  \usepackage[nocompress]{cite}
\else
  \usepackage{cite}
\fi
\ifCLASSINFOpdf
\else
\fi
\usepackage{amsmath}
\usepackage{color}
\usepackage{graphicx}
\usepackage{threeparttable}
\usepackage[linesnumbered, ruled, vlined]{algorithm2e}
\usepackage{amsmath}
\usepackage{amssymb}
\usepackage{amsthm}
\usepackage{mathtools}
\usepackage{soul}
\usepackage[normalem]{ulem}
\usepackage{times}

\newtheorem{definition}{\textbf{Definition}}[section]
\newtheorem{proposition}{\textbf{Proposition}}[section]
\newtheorem{assumption}{Assumption}
\newtheorem{principle}{\textbf{Principle}}[section]

\newtheorem{property}{\textbf{{Property}}}[section]

\usepackage[noend]{algpseudocode}
\usepackage{indentfirst}

\makeatletter
\newcommand{\rmnum}[1]{\romannumeral #1}
\newcommand{\Rmnum}[1]{\expandafter\@slowromancap\romannumeral #1@}
\makeatother

\usepackage{array}
\newcolumntype{C}[1]{>{\centering\arraybackslash}m{#1}}

\definecolor{azure(colorwheel)}{rgb}{0.0, 0.5, 1.0}

\definecolor{frenchblue}{rgb}{0.0, 0.45, 0.73}

\definecolor{bittersweet}{rgb}{1.0, 0.44, 0.37}

\definecolor{green(pigment)}{rgb}{0.0, 0.65, 0.31}

\definecolor{navyblue}{rgb}{0.0, 0.0, 0.5}

\definecolor{darkcerulean}{rgb}{0.03, 0.27, 0.49}

\definecolor{darkpowderblue}{rgb}{0.0, 0.2, 0.6}

\definecolor{egyptianblue}{rgb}{0.06, 0.2, 0.65}

\begin{document}
%
\title{Towards Designing Cost-Optimal Policies to Utilize IaaS Clouds with Online Learning}
%
%
%
%

\author{Xiaohu~Wu,
        Patrick~Loiseau,
        and~Esa Hyyti\"a
\IEEEcompsocitemizethanks{\IEEEcompsocthanksitem Xiaohu Wu is with Fondazione Bruno Kessler, Trento,
Italy. 
E-mail: xiaohuwu@fbk.eu
\IEEEcompsocthanksitem Patrick Loiseau is with Univ. Grenoble Alpes, Inria, CNRS, Grenoble INP, LIG, France and MPI-SWS, Germany. 
E-mail: patrick.loiseau@inria.fr

\IEEEcompsocthanksitem Esa Hyyti\"a is with University of Iceland, Reykjav\'ik, Iceland. 
E-mail: esa@hi.is

}
\thanks{Manuscript received April 19, 2005; revised August 26, 2015.}}

%
%

\markboth{Journal of \LaTeX\ Class Files,~Vol.~14, No.~8, August~2015}%
{Shell \MakeLowercase{\textit{et al.}}: Bare Demo of IEEEtran.cls for Computer Society Journals}
%



\IEEEtitleabstractindextext{%
\begin{abstract}
Many businesses possess a small infrastructure that they can use for their computing tasks, but also often buy extra computing resources from clouds. Cloud vendors such as Amazon EC2 offer two types of purchase options: on-demand and spot instances. As tenants have limited budgets to satisfy their computing needs, it is crucial for them to determine how to purchase different options and utilize them (in addition to possible self-owned instances) in a cost-effective manner while respecting their response-time targets. In this paper, we propose a framework to design policies to allocate self-owned, on-demand and spot instances to arriving jobs. In particular, we propose a near-optimal policy to determine the number of self-owned instances and an optimal policy to determine the number of on-demand instances to buy and the number of spot instances to bid for at each time unit. Our policies rely on a small number of parameters and we use an online learning technique to infer their optimal values. Through numerical simulations, we show the effectiveness of our proposed policies, in particular that they achieve a cost reduction of up to 64.51\% when spot and on-demand instances are considered and of up to 43.74\% when self-owned instances are considered, compared to previously proposed or intuitive policies.
\end{abstract}

}

\maketitle

\IEEEdisplaynontitleabstractindextext

%
\IEEEpeerreviewmaketitle

\section{Introduction}

Infrastructure-as-a-Service (IaaS) holds exciting potential for users to elastically scale their computation capacity up and down to match the time-varying demand. This eliminates the need of purchasing their own servers to satisfy the peak demand, without causing unacceptable latencies. The global cloud IaaS market grew to \$34.6 billion in 2017, and is projected to increase to \$71.6 billion in 2020 \cite{Gartner}. Main IaaS service providers include Amazon, Microsoft, Google, etc. Amazon is the most popular and represents 51.8\% of the global market share in 2017; here, two typical purchase options are on-demand and spot instances (i.e., virtual machines). Recently, the issue of cost-effectively utilizing these two types of instances has received significant attention \cite{Kumar18a}.

On-demand instances are always available at a fixed price and tenants\footnote{In this paper, we use "users" and "tenants" interchangeably.} pay only for the period in which instances are consumed at an hourly rate. Users can also bid a price for spot instances and successfully get them only if their bid is above the spot price. However, spot instances will get lost once the spot price becomes higher than their bid. Here, spot prices usually vary unpredictably over time and users will be charged the spot prices for their use. Compared to on-demand instances, spot instances can reduce the cost by up to 50-90\% \cite{aws}. Users that purchase cloud instances may also have their own instances, referred to as self-owned instances, which can be used to process jobs but are insufficient at times (hence the need to purchase extra IaaS instances). They may also not have any self-owned instances (e.g., in the case of startups) and therefore need to buy from the cloud all necessary computing resources.

Tenants' jobs arrive over time. We focus on processing a type of embarrassingly parallel workloads/jobs \cite{Fox11a,Gunarathne10}. Each job can be separated into a large number of small tasks. These tasks are independent and can be executed on multiple machines simultaneously. Completing a job means completing all its tasks and the maximum completion time of all tasks is the job's completion time. This type of jobs accounts for a significant proportion in cloud market; examples include 3D video rendering, BLAST searches, data cleaning and pre-processing. Such job is also called malleable job \cite{Jain12,Nagarajan13a,Wu15a,Wu15b} and it has a parallelism bound specifying the maximum number of instances that it can utilize simultaneously. Each job also has constraint on timing, i.e., a deadline by which to complete all tasks of a job. Subject to the parallelism constraint, an arriving job will be allocated instances of different types (self-owned, on-demand and spot) and the allocation can be updated at most once every hour (since billing is done per hour). Our problem is then to find an allocation that minimizes cost while satisfying the deadline constraint.

\vspace{0.08em}\noindent\textbf{Challenges.} In this paper, we make a natural assumption that the costs of utilizing self-owned, spot and on-demand instances are increasing. To be cost-optimal, an allocation policy should sequentially maximize the utilization of self-owned and then spot instances. This is, however, a difficult task. For instance, {\em a naive policy} would be, whenever a job arrives, to assign as many remaining self-owned instances as possible to it. However, this policy turns out not to be good wrt cost. Indeed, it ignores the difference of jobs and treats all jobs equally when assigning instances, whereas we find that a good policy wrt cost needs instead to determine the allocations of self-owned instances to jobs according to their capabilities of utilizing spot instances. In particular, subject to the parallelism constraint of a job $j$, the availability of spot instances varies in the period between its arrival and deadline and determines the maximum workload of $j$ that could be processed by spot instances. If the workload of $j$ is large, $j$ has to utilize some stable self-owned or on-demand instances in order to finish itself by the deadline; such job is said to have poor capability of utilizing spot instances alone to finish itself by its deadline (also called poor jobs). If the workload of $j$ is small, finishing $j$ by its deadline only needs to utilize spot instance alone, with no need of self-owned or on-demand instances; such job is said to have strong such capability (called rich job).

A better policy would then be as follows. When self-owned instances are inadequate, actively assign self-owned instances to poor jobs and assign nothing to rich jobs; otherwise, such poor jobs will have to consume costly on-demand instances, and it also causes a waste of other rich jobs' capabilities to utilize spot instances.
When self-owned instances are adequate, assign a proper amount of self-owned instances to every job with either poor or strong capability such that after the allocations all jobs are expected to be completed by utilizing spot instances alone, eliminating the need of consuming costly on-demand instances. After allocating self-owned instances, the remaining question is to identify a job's capacity to utilize spot instances for processing its workload, and propose an expected optimal policy to achieve the capacities of jobs, further escaping unnecessary consumption of on-demand instances.

\vspace{0.15em}\noindent\textbf{Our Contributions.} In this paper, we propose a framework to design policies to allocate various instances. Based on the two principles that (\rmnum{1}) self-owned instances should be allocated to maximize their utilization while maximizing the opportunity of all jobs utilizing spot instances and (\rmnum{2}) on-demand instances should be allocated to maximize the opportunity to utilize spot instances, we propose parametric policies for the allocation of self-owned, on-demand and spot instances that achieve near-minimal costs. To cope with the cloud market dynamic and the uncertainty of job's characteristics, we use the online learning technique in \cite{Jain14} to infer the optimal parameters. More specifically:
\begin{itemize}
\item We propose a cost-effective policy for allocating self-owned instances that is smarter than the naive allocation mentioned above and hits a good trade-off between the utilization of self-owned instances and the opportunity of utilizing spot instances. We show in our numerical experiments that this policy improves the cost by up to 43.74\% compared to the naive policy.
\item We propose a cost-optimal policy for the utilization of on-demand and spot instances, based on a formulation of the original problem as an integer program to maximize the utilization of spot instances. This policy can be used both when the tenant has self-owned resources and when he does not. Our simulation results show that it improves the cost of previous policies in \cite{Jain14} by up to 64.51\%.
\end{itemize}
We note that the paper \cite{Jain14} also appears in \cite{Jain14P} as a U.S. Patent.

The rest of this paper is organized as follows. We introduce the related works in Section~\ref{sec.related-work} and describe the problem formally in Section~\ref{sec.model}. In Section~\ref{sec.scheduling}, we propose scheduling policies for self-owned, on-demand and spot instances. In Section~\ref{sec.evaluation}, simulations are done to show the effectiveness of the solutions of this paper. Finally, we conclude this paper in Section~\ref{sec.conclusion}. The proofs of all propositions are omitted and can be found in the appendix. We note that a part of results of this paper also appeared at the 2017 IEEE International Conference on Cloud and Autonomic Computing \cite{Wu17}.

\section{Related Work}
\label{sec.related-work}

In this paper, we use the online learning technique to learn the most effective parametric policy for utilizing various instances. Jain et al. were the first to consider the application of this approach to the scenario of cloud computing\footnote{The objective of this paper corresponds to a special case of \cite{Jain14} where the value of each job is larger than the cost of completing it.} \cite{Jain14}. However, they do not consider the problem of how to optimally utilize the purchase options in IaaS clouds and self-owned instances are also not taken into account. This approach is interesting because it does not impose the restriction of a priori statistical knowledge of workload, compared to other techniques such as stochastic programming. However, it can achieve good performances only if the potentially optimal scheduling policies are identified among all possible policies. Similar to our paper and \cite{Jain14}, cost-effectively executing deadline-constrained jobs in IaaS clouds is also studied in \cite{Song12b,Yao14a}. In particular, Zafer et al. characterize the evolution of spot prices by a Markov model and propose an optimal bidding strategy for utilizing spot instances to complete a serial or parallel job by some deadline \cite{Song12b}. Yao et al. study the problem of utilizing reserved and on-demand instances to complete online batch jobs by their deadlines and formulate it as integer programming problems; then heuristic algorithms are proposed to give approximate solutions \cite{Yao14a}.

There have been substantial works on cost-effective resource provisioning in IaaS clouds \cite{Manvi}, and we introduce some typical approaches. Built on the assumption of a priori statistical knowledge of the workload or spot prices, several techniques could be applied. For example, in \cite{Hong,Chaisiri}, the techniques of stochastic programming is applied to achieve the cost-optimal acquisition of reserved and on-demand instances; in \cite{Zheng15}, the optimal strategy for the users to bid for the spot instances are derived, given a predicted distribution over spot prices. However, a high computational complexity arises when implementing these techniques, though the statistical knowledge could be derived by the techniques such as dynamic programming \cite{Shib}.

Wang et al. use the competitive analysis technique to purchase reserved and on-demand instances without knowing the future workload \cite{Wang}, where the Bahncard problem is applied to propose a deterministic and a randomized algorithm. In \cite{Vintila13a}, a genetic algorithm is proposed to quickly approximate the pareto-set of makespan and cost for a bag of tasks where on-demand and spot instances are considered. In \cite{Shib}, the technique of Lyapunov optimization is applied and it's said to be the first effort on jointly leveraging all three common IaaS cloud pricing options to comprehensively reduce the cost of users; however, a large delay will be caused when processing jobs; in order to achieve an $\mathcal{O}(\epsilon)$ close-to-optimal performance, the queue size has to be $\Theta(1/\epsilon)$ \cite{Huang14aa}. Gao et al. proposed a two-timescale markov decision process approach by jointly considering resource provisioning and task scheduling in public clouds to maximize the profit of a multimedia service provider \cite{Gao16a}. In \cite{Dubois15a,Dubois16a}, Dubois and Casale propose a heuristic to help cloud users decide what type of spot instances should be rent and what bid price to use to minimize their cost while maintaining an acceptable level of performance.

\section{Problem Description and Model}
\label{sec.model}

In this section, we introduce the cloud pricing models, define the operational space of a user to utilize various instances, and characterize the objective of this paper.

\subsection{Pricing Models in the Cloud}
\label{sec.pricing-model}


We first introduce the pricing models in the cloud. The price of an {\em on-demand instance} is charged on an hourly basis and it is fixed and denoted by $p$. Even if on-demand instances are consumed for part of an hour, the tenant will be charged the fee of the entire hour, as illustrated in Fig.~\ref{Fig4}.
  \begin{figure}[!ht]
  \centering
  \includegraphics[width=3.25in]{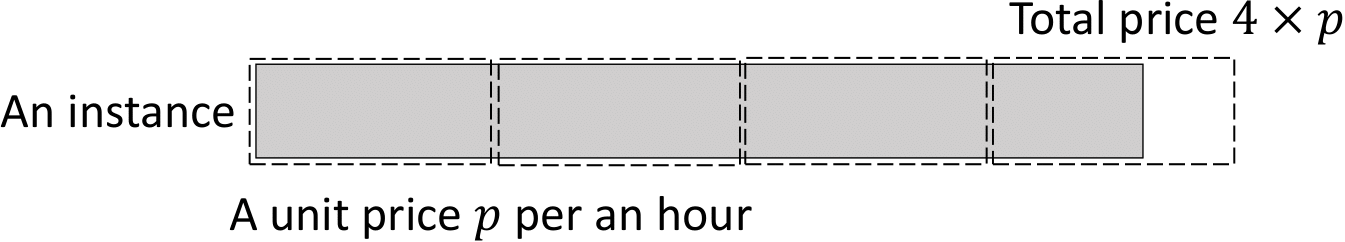}
  \caption{On-demand price: users are charged on an hourly basis at a fixed price $p$.}\label{Fig4}
 \end{figure}

Furthermore, tenants can bid a price for {\em spot instances} and {\em spot prices are updated at regular time slots} (e.g., every $L=5$ minutes in Amazon) \cite{Zheng15,Wu19a}. Spot instances are assigned to a job and continue running if the spot price is lower than the bid price. Spot prices usually change unpredictably over time \cite{Ben-Yehuda}. Once the spot price exceeds the bid price of a job, its spot instances will get terminated immediately by the cloud, as illustrated in Fig.~\ref{Fig5}; here, the termination occurs at the very beginning of a time slot. The tenant will be charged the spot prices for the maximum integer hours of execution. A partial hour of execution is not charged in the case where its instances are terminated by the cloud; in contrast, if spot instances run until a job is completed and then are terminated by the tenant, for the partial hour of execution, the tenant will also be charged for the full hour.

  \begin{figure}[!ht]
  \centering
  \includegraphics[width=3.25in]{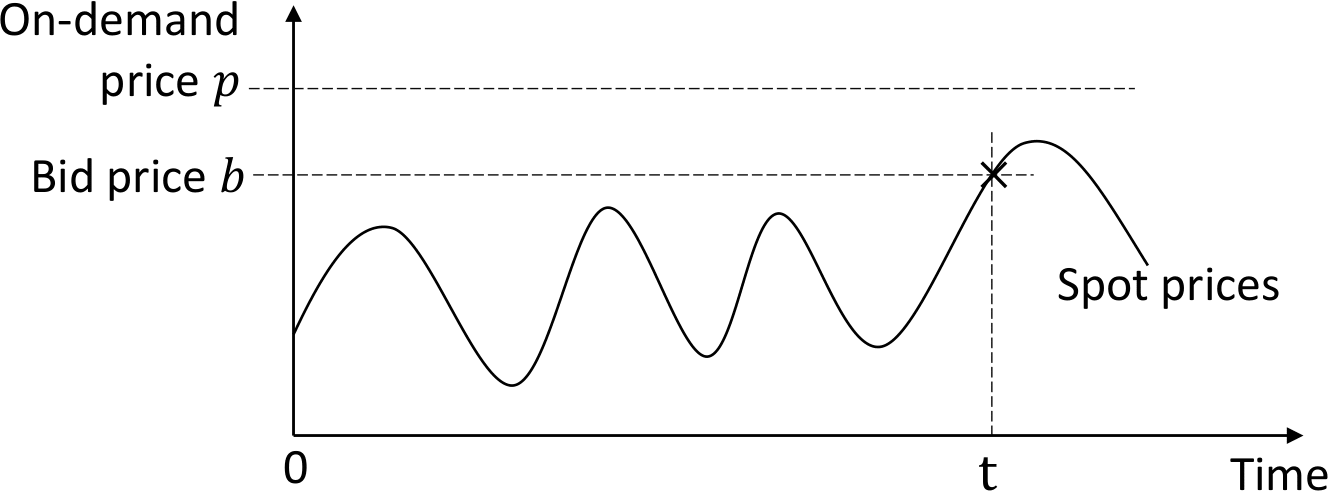}
  \caption{Spot price: a user bids a price $b$ for an instance at time $0$ and can use it until time $t$.}\label{Fig5}
  \end{figure}

Finally, a user might have its own instances, i.e., {\em self-owned instances}. The (averaged) hourly cost of utilizing self-owned instances is assumed to be $p_{1}$. It is the cheapest to use self-owned instances so that $p_{1}$ is without loss of generality assumed to be 0, which implies that a user always prefers to first utilize its own instances before purchasing instances from the cloud. An example of self-owned instances is academic private clouds, which are provided to researchers free of charge.

\subsection{Jobs}
\label{sec.jobs}

Our paper focuses on processing a type of embarrassingly parallel computations, also called map-only jobs; see \cite{Fox11a} for the classification of big-data applications on clouds. Each job can be separated into a large number of small tasks. These tasks are independent and can be executed on multiple machines simultaneously; completing a job means completing all its tasks, and the maximum completion time of all tasks is the job's completion time. Examples of such computations include parallel rendering in computer graphics, BLAST searches and CAP3 in bioinformatics, large scale facial recognition systems that compare thousands of faces, grid and random search for hyperparameter optimization in machine learning, data cleaning and pre-processing and so on. A task is the minimum running unit and should be processed continuously without preemption until its completion.

For example, in CAP3, the data are divided into many files and each task finishes reading a file; if each file has 458 reads, it may take about 7 seconds to process a task when two high CPU extra large instances are used with 8 workers per instance \cite{Gunarathne10}. For each job $j$, its size/workload is defined as the time when finishing it on one instance, denoted by $z_{j}$, and can be estimated by the input data size or the number of small tasks. While executing a job, the remaining workload can also be estimated by the remaining input data to be processed. Such jobs can be formally modeled as malleable jobs in literature \cite{Jain12,Nagarajan13a}. Each job $j$ has a parallelism bound $\delta_{j}$ that limits the maximum number of instances that it could utilize simultaneously. While executing a job, the number of instances assigned to a job could change over time.

The job arrival of a tenant is monitored every time slot of $L$ minutes (i.e., at the time points when spot prices change) and time slots are indexed by $t= 1, 2, \cdots$. Each job $j$ has four characteristics: (\rmnum{1}) {\em an arrival slot $a_{j}$}: If job $j$ arrives at a certain continuous time point in $[(t-1)\cdot L, t\cdot L)$, then set $a_{j}$ to $t$; (\rmnum{2}) {\em a relative deadline $d_{j}\in\mathcal{Z}^{+}$}: it is a time constraint on completing a job, that is, every job must be completed at or before time slot $a_{j}+d_{j}-1$; (\rmnum{3}) {\em a job size $z_{j}$} (measured in the instance time slots (CPU time) that need to be utilized): the workload to complete $j$; (\rmnum{4}) {\em a parallelism bound $\delta_{j}$}: the upper bound on the number of instances that could be simultaneously utilized by $j$. The tenant plans to rent instances in IaaS clouds to process its jobs and aims to minimize the cost of completing a set of jobs $\mathcal{J}$ (that arrive over a time horizon $T$) by their deadlines.

\subsection{Rules for Allocating Instances to Jobs}
\label{sec.rules}


The pricing models define the rules of allocating instances to jobs and also the operational space of a user, i.e., (a) when the allocation to a job is done and updated, and, (b) how various instances are utilized by a job at every allocation update.

\subsubsection{On-demand and spot instances}
\label{sec.on-demand-spot}

We first consider the allocation of on-demand and spot instances alone, ignoring self-owned instances temporarily.

To meet deadlines, (\textbf{\rmnum{1}}) {\em whenever a job $j$ arrives at $a_{j}$, the allocation of spot and on-demand instances to it is done immediately}. The following rules apply to the case where $j$ has flexibility to utilize spot instances. Given the fact that the tenant is charged on hourly boundaries, (\textbf{\rmnum{2}}) {\em the allocation of on-demand and spot instances to each job $j$ is updated simultaneously every hour.} {\em The $i$-th allocation} occurs at the beginning of slot $t=a_{j}+(i-1)\cdot Len$ where $Len = 60/L$ is the number of slots per hour; the number of on-demand instances allocated to $j$ is denoted by $o_{j}^{i}$ and they can be utilized for the entire hour; (\textbf{\rmnum{3}}) {\em the tenant will bid a price $b_{j}^{i}$ for a fixed number $si_{j}^{i}$ of spot instances.} At the $i$-th allocation of $j$, $b_{j}^{i}$ together with the spot prices determines whether $j$ can successfully obtain spot instances and how long it can utilize them. Usually, spot instances are on average cheaper than on-demand instances, and (\textbf{\rmnum{4}}) {\em at every allocation the tenant will bid for the maximum number of spot instances under the parallelism constraint, i.e., $si_{j}^{i}=\delta_{j}-o_{j}^{i}$.}
While $j$ is utilizing the instances of the $i$-th allocation in the period of $[a_{j}+(i-1)\cdot Len, a_{j}+i\cdot Len-1]$, we say that $j$ is in {\em the $i$-th execution}.

At the $i$-th allocation, we use $z_{j}^{i}$ to denote the remaining workload of $j$ to be processed, i.e., $z_{j}$ minus the workload of $j$ that has been processed. We define the current slackness of $j$ as
\begin{align}
s_{j}^{i}=\left( d_{j}-(i-1)\cdot Len \right)\cdot \delta_{j}/z_{j}^{i}.
\end{align}
The slackness can be used to measure the time flexibility that $j$ has to utilize spot instances.
\begin{definition}\label{def-2}
During the $i$-th execution where $i\geq 1$, if spot instances get lost at the beginning of some slot $t^{\prime}$ and are not utilized for the entire hour, we say that, at the next allocation,
\begin{enumerate}
\item $j$ has flexibility to utilize spot instances, if $s_{j}^{i+1} \geq 1$;
\item $j$ does not have such flexibility, otherwise.
\end{enumerate}
Here, we know the values of $s_{j}^{i+1}$ and $z_{j}^{i+1}$ since we then know the value of $z_{j}^{i}$ and the workload processed in the $i$-th execution.
\end{definition}

In Definition~\ref{def-2}, if $s_{j}^{i+1}<1$ and the ($i+1$)-th allocation is still taken at slot $a_{j}+i\cdot Len$, we have by the definition of $s_{j}^{i+1}$ that $j$ cannot be completed by its deadline even if $j$ totally utilize $\delta_{j}$ on-demand instances in the remaining period; we use $i_{j}$ to index the last (or such $i$-th) allocation after which $j$ has no flexibility to utilize spot instances. We illustrate this by Fig.~\ref{Fig.1aa}. At the first allocation, $z_{j}^{1}=z_{j}=132$ and $o_{j}^{1}=si_{j}^{1}=2$. At the second allocation, $z_{j}^{2}=132-2\cdot 12-2\cdot 8=92$, and $o_{j}^{2}$ and $si_{j}^{2}$ are still 2. In the second execution, when spot instances are terminated at the end of slot 20, we have $z_{j}^{3}=92-2\cdot 12-2\cdot 8=52$, and $s_{j}^{3}= \frac{Len\cdot \delta_{j}}{z_{j}^{3}} < 1$; thus, there is no flexibility for $j$ to utilize spot instances at the third allocation. In Fig.~\ref{Fig.1aa}, $i_{j}=2$. The decision on how to determine the ($i_{j}+1$)-th allocation of instances to $j$ has to be done earlier (than the beginning of slot $a_{j}+i_{j}\cdot Len =25$), since there exists an on-demand instance that has to be utilized for $\frac{4}{3}$ hours to satisfy the deadline constraint.

\begin{figure}[!ht]
  \centering
  	\includegraphics[width=3.2in]{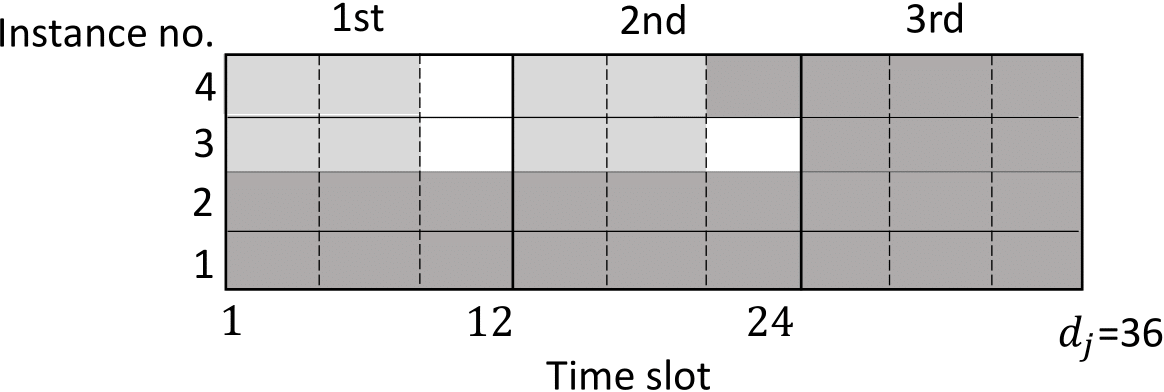}
	\caption{Illustration of the process of allocating instances to $j$ where $a_{j}=1$, $d_{j}$ equals 3 hours, $L=5$ minutes, $z_{j}=132$, and $\delta_{j}=4$: the light gray (resp. heavy gray) area in every period of $[12\cdot(i-1)+1, 12\cdot i]$ illustrates the workload processed by spot (resp. on-demand) instances in the $i$-th execution where $i=1, 2$.}
\label{Fig.1aa}
\end{figure}

Finally, the instance allocation is generally divided into two phases. In the first phase, $j$ has flexibility to utilize spot instances at every $i$-th allocation and
\begin{itemize}
\item the instance allocation is updated every hour (i.e., every $Len$ slots).
\end{itemize}
In the $i$-th execution of $j$, the workload processed by on-demand instances is $Len\cdot o_{j}^{i}$; as time goes by, there are three possible states while utilizing spot instances:
\begin{itemize}
\item [(\rmnum{1})] $z_{j}^{i}-Len\cdot o_{j}^{i}$ workload of $j$ has been processed by spot instances;
\item [(\rmnum{2})] $j$ utilizes spot instances for $Len$ slots;
\item [(\rmnum{3})] the spot instances of $j$ are interrupted by the cloud and utilized for partial hour.
\end{itemize}
With the first state, the spot instances will be terminated by its user and $j$ will be finally completed in the $i$-th execution. With the second state, since $j$ has flexibility for spot instances at the $i$-th allocation (i.e., $s_{j}^{i}\geq 1$), so does it at the ($i+1$)-th allocation. With the third state, if $s_{j}^{i+1}\geq 1$, the ($i+1$)-th allocation of $j$ is still in the first phase; otherwise,
\begin{itemize}
\item the ($i+1$)-th (i.e., ($i_{j}+1$)-th) allocation of $j$ needs to be done immediately after the spot instances get lost,
\end{itemize}
which is referred to as the second phase of instance allocation; then, only stable on-demand instances will be used to meet the deadline.

\subsubsection{Self-owned instances}
\label{sec.self-owned}

When self-owned instances are also taken into account, we assume that (\textbf{\rmnum{5}}) {\em the allocation of self-owned instances to a job can be updated at most once at every allocation of $j$}. We denote by $r_{j}^{i}$ the number of self-owned instances assigned to $j$ at the $i$-th allocation; the parallelism constraint further translates to $o_{j}^{i}+si_{j}^{i}+r_{j}^{i} = \delta_{j}$. In this paper, $o_{j}^{i}$ and $si_{j}^{i}$ denotes the numbers of on-demand and spot instances acquired at the $i$-th allocation and will be used to track the cost of completing $j$. As we will see in Section~\ref{sec.spot}, the acquired on-demand instances may not be fully utilized for an entire hour in the $(i_{j}+1)$-th execution, and, we use $o_{j}(t)$, $si_{j}(t)$ and $r_{j}(t)$ to denote the numbers of on-demand, spot and self-owned instances that are actually utilized by $j$ at every slot $t\in[a_{j}, a_{j}+d_{j}-1]$, where $r_{j}(t)=r_{j}^{i}$ for all $t\in [a_{j}+(i-1)\cdot Len, a_{j}+i\cdot Len-1]$; then the parallelism constraint translates to $o_{j}(t)+si_{j}(t)+r_{j}(t) = \delta_{j}$.

As shown later, allocating properly self-owned instances enables escaping unnecessary consumption of on-demand instances that are more expensive than the others, which can be achieved by simply allocating $j$ the same number of self-owned instances at every time slot, i.e., $r_{j}(t)=r_{j}$. So, the allocation of self-owned instances is done only once upon arrival of $j$; after the allocation, the job can could be viewed as a new job with a parallelism bound $\delta_{j}-r_{j}$, a size $z_{j}-r_{j}\cdot d_{j}$, and the same arrival time and deadline as $j$ , and it will be completed by utilizing spot and on-demand instances alone.

\subsection{Scheduling Objective}
\label{sec.objective}

We refer to the ratio of the total cost of utilizing a certain type of instances to the total workload processed by this type of instances as the average unit cost of this type of instances. As described in Section~\ref{sec.pricing-model}, we assume that
\begin{assumption}\label{assump-order}
The average unit costs of self-owned instances is lower than the average unit cost of spot instances, which is lower than that of on-demand instances. 
\end{assumption}
Accordingly, to be cost-optimal, we should consider allocating various instances to each arriving job in the order of self-owned, spot and on-demand instances. Further, in Principles~\ref{prin-one} and~\ref{prin-two}, we give the objectives that should be achieved when considering allocating each type of instances to the arriving jobs.

\begin{principle}\label{prin-one}
The scheduler should make self-owned instances (\rmnum{1}) fully utilized, and (\rmnum{2}) utilized in a way so as to maximize the opportunity that all jobs have to utilize spot instances.
\end{principle}

\begin{principle}\label{prin-two}
After self-owned instances are used, the scheduler should utilize on-demand instances in a way so as to maximize the opportunity that all jobs have to utilize spot instances.
\end{principle}

\subsubsection{Decision variables}

Our main objective of this paper is to propose scheduling policies that can realize Principles~\ref{prin-one} and~\ref{prin-two}. To do so, we will first determine the allocation of self-owned instances and then the allocation of on-demand and spot instances for every arriving job $j$. For every arriving job $j$, it will be first allocated $r_{j}$ self-owned instances in $[a_{j}, a_{j}+d_{j}-1]$. Then, as described in Section~\ref{sec.rules}, the allocation of spot and on-demand instances will be updated per hour in the first phase and we need to determine the number of spot instances to be bid for and the number of on-demand instances to be purchased (i.e., $si_{j}^{i}$ and $o_{j}^{i}$); once there is no flexibility for $j$ to utilize spot instances, we need to determine the allocation of on-demand instances alone in order to complete $j$ by deadline. Hence, the main decision variables of this paper are $r_{j}$, $si_{j}^{i}$, and $o_{j}^{i}$ where $o_{j}^{i}+si_{j}^{i}+r_{j} = \delta_{j}$.

In this paper, we apply the online learning approach and it does not require the exact statistical knowledge on jobs and spot prices. At every allocation update of $j$ in the first phase, only the current characteristics of $j$ (i.e., $z_{j}^{i}$, $\delta_{j}$, $a_{j}$, and $d_{j}$) and the amount of available self-owned instances are definitely known for a user to determine $o_{j}^{i}$ and $si_{j}^{i}$. The value of spot price is jointly determined by the arriving jobs of numerous users and the number of idle servers at a moment, usually varying over time unpredictably. In this paper, it is assumed that the change of spot prices over time is independent of the job's arrival of an individual user \cite{Song12b,Zheng15}. In the $i$-th execution of $j$, when a user bids some price for $si_{j}^{i}$ spot instances, without considering the case where the spot instances of $j$ is terminated by a user itself, the period in which $j$ can utilize spot instances is a random variable and we assume that the expected time for which $j$ could utilize spot instances is $\beta\cdot Len$ where $\beta\in [0,1]$.
Finally, Table~\ref{table} summarizes the main notation of this paper.

\begin{table}
	\centering
	\begin{threeparttable}[t]
		\caption{Main Notation}
		\begin{tabular}{| C{1.55cm} | C{6.2cm} |}
			\hline
			Symbol & Explanation\\ \hline
			
			$L$ & length of a time slot (e.g., 5 minutes) \\ \hline
			
			$Len$ & the number of time slots in an hour, i.e., $\frac{60}{L}$ \\ \hline
			
			$\mathcal{J}$ & a set of jobs that arrive over time \\ \hline
			
			$j$ and $a_{j}$ & a job of $\mathcal{J}$ and its arrival time \\ \hline
			
			$d_{j}$ & the relative deadline: $j$ must be completed by a deadline $a_{j}+d_{j}-1$ \\ \hline
			
			$z_{j}$ & the job size of $j$, measured in CPU $\times$ time slots \\ \hline
			
			$\delta_{j}$ & the parallelism bound, i.e., the maximum number of instances that can be simultaneously used by $j$ \\ \hline
			
			$T$ & the number of time slots, i.e., $\max_{j\in\mathcal{J}}\{ a_{j} \}$ \\ \hline

			$p$ and $p_{1}$ & the price of respectively using an on-demand and self-owned instance for an hour \\ \hline

			$si_{j}^{i}$, $b_{j}^{i}$, and $o_{j}^{i}$ & the number of spot instances bid for, the bid price, and the number of on-demand instances acquired at the $i$-th allocation of $j$  \\ \hline

			$r_{j}(t)$, $si_{j}(t)$ and $o_{j}(t)$ & the number of self-owned, spot and on-demand instances utilized by $j$ at a slot $t$  \\ \hline
			
			
			$z_{j}^{i}$  & the remaining workload of $j$ to be processed at the $i$-th allocation of $j$  \\ \hline
			
			$s_{j}^{i}$ & the slackness at the $i$-th allocation, i.e., $(d_{j}-(i-1)\cdot Len)\cdot \delta_{j}/ z_{j}^{i}$  \\ \hline

            $i_{j}$ & the last allocation of $j$ at which $j$ has flexibility to utilize spot instances \\ \hline	
			
			
			$r_{j}$ & the number of self-owned instances allocated to a job $j$ at every $t\in [a_{j}, a_{j}+d_{j}-1]$ \\ \hline
			
			$\beta$ & the availability of spot instances varies over time; at every allocation, the expected time for which a job could utilize spot instances is $\beta\cdot Len$ where $\beta\in [0,1]$   \\ \hline

			$R$ & the number of self-owned instances\\ \hline

            $N(t)$ & the number of currently idle self-owned instances at a slot $t$ \\ \hline

            $m_{t_{1}}(t_{2})$ & the maximum number of self-owned instances idle at every slot in $[t_{1}, t_{2}]$, i.e., $\min\left\{ N(t_{1}), \cdots, N(t_{2}) \right\}$ \\ \hline

			$b$ & the bid price   \\ \hline

			$\beta_{0}$ & a parameter that control the allocation of self-owned instances via Equation (4)  \\ \hline

			$\{\beta, \beta_{0}, b\}$ &  a parameterized policy for allocating various instances to a job at every allocation, as stated in the Section~\ref{sec.framework}   \\ \hline
			
			$\mathcal{P}$ & a set of parameterized policies    \\ \hline

            $\pi$ & the index of a policy in $\mathcal{P}$: $\pi=1, 2, \cdots$ \\ \hline					
			
		\end{tabular}
		\label{table}
	\end{threeparttable}
\end{table}

\section{The Design of Near-Optimal Policies}
\label{sec.scheduling}

In this section, we propose a theoretical framework to design (near-)optimal parametric policies, aiming at realizing Principles~\ref{prin-one} and~\ref{prin-two}. Facing diverse users, the proposed policies should have good adaptability against the unknown statistics of the spot prices and each individual user's job characteristics; then, by applying the online learning technique, the best configuration parameter that corresponds to each user could be inferred to minimize its cost of processing jobs.

Upon arrival of a job $j$, the scheduler first considers the allocation of self-owned instances to it, aiming to realize the two goals in Principle~\ref{prin-one}. Next, as described in Section~\ref{sec.on-demand-spot}, the allocation of spot and on-demand instances is updated on an hourly basis.

\subsection{Self-owned Instances}
\label{sec.long-term}

In this subsection, we study the allocation of self-owned instances.

\subsubsection{Challenge}

We first show the challenges in cost-effectively utilizing self-owned instances by an example. Initially, there is a fixed number $R$ of self-owned instances. Upon arrival of a job $j$, it is allocated a fixed number of self-owned instances, and these instances will be reserved for this job in the period $[a_{j}, a_{j}+d_{j}-1]$ and released by $j$ after the slot $a_{j}+d_{j}-1$. As time goes by, when time is at the beginning of any slot $t$, we have the information on the allocation of self-owned instances to the previous jobs (i.e., the number of self-owned instances allocated to each previous job and the period in which these instances are reserved for this job) and we can get the number of self-owned instances that are not reserved for the previous jobs in the period of each slot $t^{\prime}$, denoted by $N(t^{\prime})$. Let $m_{t_{1}}(t_{2}) = \min\left\{ N(t_{1}), \cdots, N(t_{2}) \right\}$, where $t_{1}\leq t_{2}$, and it represents the maximum number of self-owned instances idle/non-reserved at every slot in $[t_{1}, t_{2}]$. \textbf{An intuitive policy} would be, whenever a job $j$ arrives, to allocate as many self-owned instances to $j$ to make self-owned instances fully utilized, i.e.,
\begin{equation}\label{intuitive}
r_{j} = \min\{ m_{a_{j}}(a_{j}+d_{j}-1), z_{j}/d_{j} \}.
\end{equation}
However, this intuitive policy may not maximize the opportunity that all jobs have to utilize spot instances as illustrated in the following example.

There are two self-owned instance available, and two jobs whose have the same arrival time, relative deadline of 2 hours and parallelism bound of 4. Jobs 1 and 2  have a size of $4\times Len$ and $6\times Len$, respectively. It is expected that a job can utilize spot instances for $\beta=\frac{1}{2}$ hour ($\beta\cdot Len$ slots) at every allocation update. In Fig.~\ref{Fig.031}, the area of diagonal stripes, the area of vertical stripes, and the dotted area denote the workload respectively processed by spot, self-owned and on-demand instances. Using the policy (\ref{intuitive}), the whole process of allocating instances is illustrated in Fig.~\ref{Fig.031} (left), where the user has to utilize two on-demand instances for 0.5 hour; however, it is not necessary to purchase more expensive on-demand instances if the allocation process is like Fig.~\ref{Fig.031} (right). In Fig.~\ref{Fig.031} (left), the cost of completing jobs 1 and 2 is $2\cdot p$ while it is zero in Fig.~\ref{Fig.031}; here, on-demand instances are charged on an hourly basis, and the fee of utilizing spot instances is zero when they are terminated by the cloud.

\begin{figure}[!ht]
	\centering

	\includegraphics[width=3.15in]{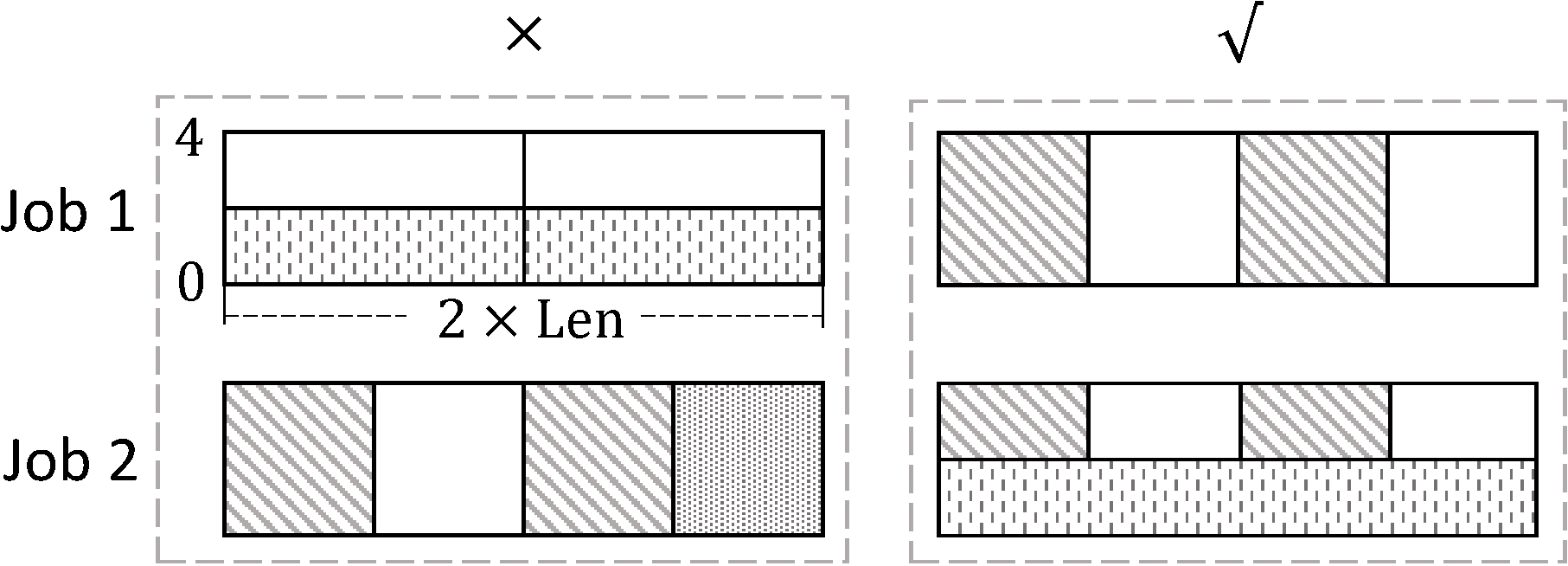}\\
	
	\caption{The Challenge in Cost-Effectively Utilizing Self-owned Instances.}\label{Fig.031}
\end{figure}

The above example reveals some challenges in designing cost-effective policies for allocating self-owned instances. For example, the policy should have the ability of (\rmnum{1}) identifying the subset of jobs that can be expected to be completed by utilizing spot instances alone even if they are not allocated any self-owned instance, e.g., the job 1 in Fig.~\ref{Fig.031} (right), and (\rmnum{2}) properly allocating self-owned instances to the rest of jobs, when self-owned instances are inadequate. All in all, our aim is to realize Principle~\ref{prin-one}.

\subsubsection{Policy Design}

In the following, we propose a policy that has the abilities described above. In the subsequent analysis, the issue of rounding the allocations of a job to integers is ignored temporarily for simplicity; in reality, we could round the allocations up to integers, which does not affect the related conclusions much as shown by the analysis.

Recall the meaning of $\beta$ in Section~\ref{sec.objective}. For every job $j$, we will go to find a function $g_{j}(x)\in [0, \frac{z_{j}}{d_{j}}]$ that satisfies the following properties where $\frac{z_{j}}{d_{j}}\leq \delta_{j}$:

\begin{property}\label{proper-2}
$g_{j}(x)$ is non-increasing as $x$ increases in $[0, 1)$.
\end{property}

\begin{property}\label{proper-1}
$g_{j}(\beta)$ is the minimum number such that when a job $j$ is assigned $r_{j}$ self-owned instances in $[a_{j}, a_{j}+d_{j}-1]$ where $r_{j} \geq g_{j}(\beta)$, it could be expected that
\begin{itemize}
\item job $j$ could be completed by its deadline by utilizing spot instances alone if $\delta_{j}-r_{j}$ spot instances are bid for at every allocation update of $j$, where no on-demand instances is acquired.
\end{itemize}
\end{property}

The value of $g_{j}(\beta)$ is an indicator of the capability that $j$ has such that it can be completed by utilizing spot instances alone. By Property~\ref{proper-1}, if $g_{j}(\beta)\leq 0$, it is expected that no self-owned or on-demand instances is needed in order to complete $j$ and such jobs have strong capability to feed themselves with spot instances. Otherwise, $g_{j}(\beta)$ self-owned instances are needed, or $j$ has to consume some amount of expensive on-demand instances in order to be completed by its deadline; for a job $j$, the larger the value of $g_{j}(\beta)$, the weaker its capability to feed itself with spot instances.

Let $\kappa_{0} = \lceil \frac{d_{j}}{Len} \rceil-1$, and we set
\[
\overline{r}_{j}(x) =
\begin{cases}
& r_{j}^{\prime}(x) \quad \text{ if } d_{j}-\kappa_{0}\cdot Len > x\cdot Len, \\
&r_{j}^{\prime\prime}(x) \quad \text{ if } d_{j}-\kappa_{0}\cdot Len \leq x\cdot Len,
\end{cases}
\]
where
\begin{align*}
r_{j}^{\prime}(x) = \delta_{j} - \frac{d_{j}\cdot \delta_{j}-z_{j}}{d_{j}-(\kappa_{0}+1)\cdot Len\cdot x},
\end{align*}
and
\[
r_{j}^{\prime\prime}(x) =
\begin{cases}
& 0\quad\quad\quad\quad\quad\quad\quad\quad \text{ if } \kappa_{0}=0, \\
&\delta_{j} - \frac{d_{j}\cdot\delta_{j}-z_{j}}{(1-x)\cdot\kappa_{0}\cdot Len}\quad \text{ if } \kappa_{0}\geq 1.
\end{cases}
\]

We further set
\begin{align}\label{fun-self-owned}
g_{j}(x)=\max\left\{\overline{r}_{j}(x), 0\right\}.
\end{align}
When $x=0$, $g_{j}(x)=\max\{ r_{j}^{\prime}(x), 0 \}=\frac{z_{j}}{d_{j}}$. When $x\rightarrow 1$, $g_{j}(x)=\max\{ r_{j}^{\prime\prime}(x), 0 \}$ and we have that (\rmnum{1}) if $\kappa_{0}=0$, $g_{j}(x)=0$, (\rmnum{2}) if $\kappa_{0}\geq 1$ and $d_{j}\cdot\delta_{j}=z_{j}$, $g_{j}(x)=\frac{z_{j}}{d_{j}}$, and (\rmnum{3}) if $\kappa_{0}\geq 1$ and $d_{j}\cdot\delta_{j}>z_{j}$, $g_{j}(x)= 0$ since $r_{j}^{\prime\prime}(x)\rightarrow -\infty$. Now, we proceed to show that the particular $g_{j}(x)$ in (\ref{fun-self-owned}) satisfies Properties~\ref{proper-1} and~\ref{proper-2}.

\begin{proposition}\label{proposi-long}
The function $g_{j}(x)$ in (\ref{fun-self-owned}) satisfies Property~\ref{proper-1}.
\end{proposition}

\begin{proposition}\label{proposi-long-1}
The function $g_{j}(x)$ in (\ref{fun-self-owned}) satisfies Property~\ref{proper-2}.
\end{proposition}

In this paper, we consider a set of jobs $\mathcal{T}$ that arrive over time and can have diverse characteristics. When $x$ ranges in $[0, 1)$, we illustrate the function $g_{j}(x)$ in Fig.~\ref{Fig.00310} where $z_{j}=240$, $L=5$, $\delta_{j}=20$, and $Len=12$. The job's minimum execution time is $\frac{z_{j}}{\delta_{j}}=Len$ where $j$ is assigned $\delta_{j}$ instances throughout its execution. The job's deadline reflects its ability to utilize spot instances and in Fig.~\ref{Fig.00310} the solid curves from left to right represent $g_{j}(x)$ where $d_{j}$ is respectively
$5$, $3$, $2.1$, $1.47$, $1.25$, $1.11$, and $1.02$ times $Len$: under the same $x$, the larger the deadline, the smaller the value of $g_{j}(x)$. Given $z_{j}$, $\delta_{j}$ and $d_{j}$, we can see in Fig.~\ref{Fig.00310} that the function $g_{j}(x)$ is non-increasing as $x$ ranges in $[0, 1)$.

\begin{figure}[!ht]
	\centering
	
	\includegraphics[width=3.05in]{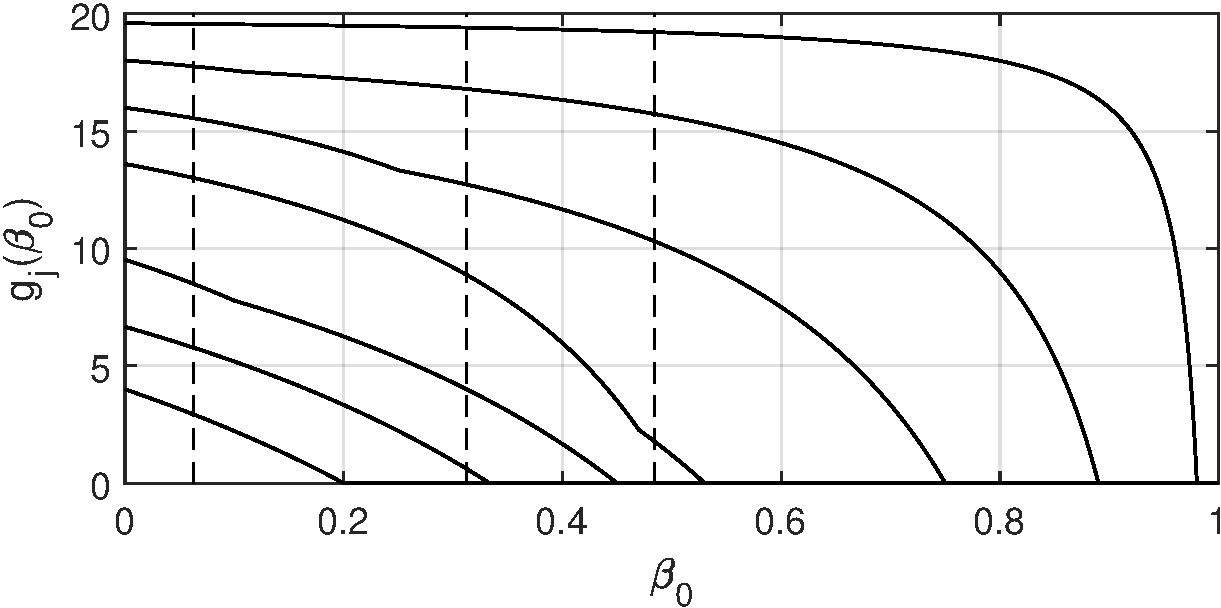}\\
	
	\caption{As $x$ ranges in $[0, 1)$, the function $g_{j}(x)$ for jobs respectively with different flexibility to utilize spot instances.}\label{Fig.00310}
\end{figure}

\vspace{0.20em}\noindent\textbf{Proposed Policy.} Based on Propositions~\ref{proposi-long} and \ref{proposi-long-1}, we propose the following policy for allocating self-owned instances. Upon arrival of every job $j$, it is allocated $r_{j}(\beta_{0})$ self-owned instances where
\begin{equation}\label{policy-long-term}
 r_{j}(\beta_{0}) =  \min\left\{g_{j}(\beta_{0}), m_{t}(a_{j}+d_{j}-1) \right\},
\end{equation}
where $\beta_{0}\in [0,1)$ is a parameter to be learned.

This policy achieves more cost-effective resource allocation as illustrated in Fig.~\ref{Fig.031} (right) where $\beta_{0}$ is set to $\beta=\frac{1}{2}$. Furthermore, this policy is also adaptive. For example, given another user who owns more instances (e.g., 5 instances), $\beta_{0}$ can be set to a value $<\beta$ (e.g., 0); then, both jobs are allocated more self-owned instances: $r_{1}=2$, and $r_{2}=3$. As a result, self-owned instances are fully utilized and there is no need purchasing spot or on-demand instances.

\subsubsection{Explanation}
\label{sec.explanation}

Now, we further explain that the policy (\ref{policy-long-term}) has desired properties to realize Principle~\ref{prin-one}, which will also be validated by the simulations.

\begin{figure}[!ht]
	\centering
	
	\includegraphics[width=3.3in]{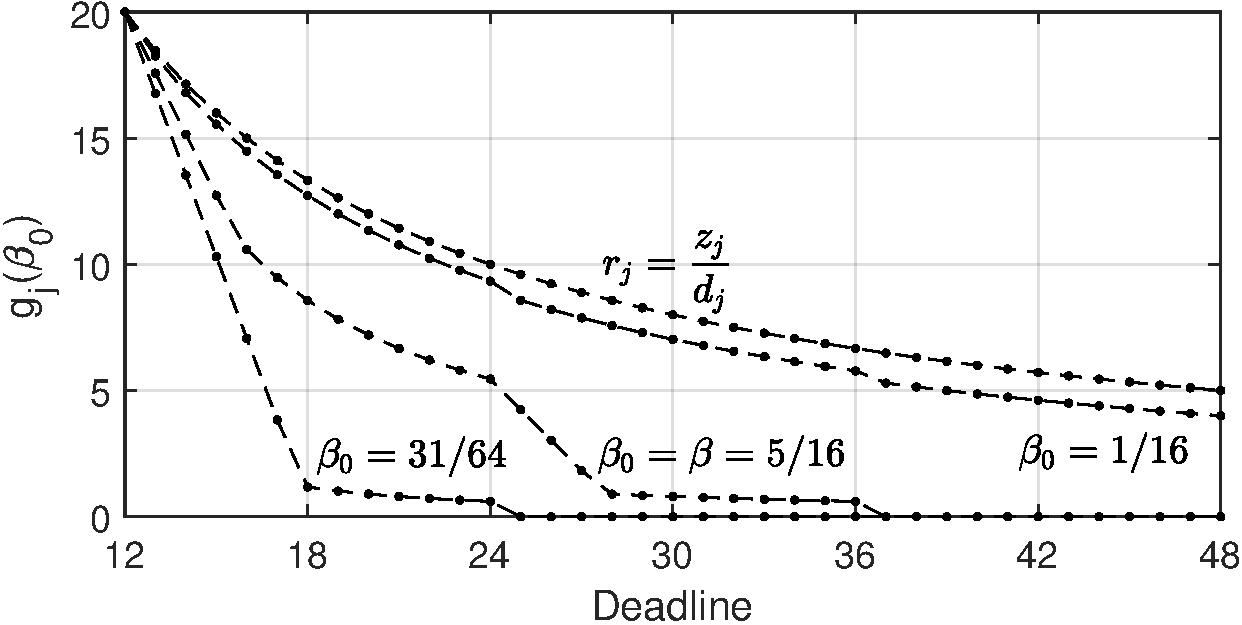}\\
	
	\caption{As the (relative) deadline $d_{j}$ increases from 12 to 48, the function $g_{j}(\beta_{0})$ decreases respectively under $\beta_{0}=\frac{31}{64}$, $\frac{5}{16}$, $\frac{1}{16}$, where $z_{j}=240$, $\delta_{j}=20$, and $Len=12$.}\label{Fig.0031}
\end{figure}

The allocations of self-owned instances to all jobs are based on the same function (\ref{policy-long-term}) whose value depends on a single parameter $\beta_{0}$. Together with Properties~\ref{proper-1} and~\ref{proper-2}, the power of the proposed policy can be achieved by setting $\beta_{0}$ to a value properly small in $[0, 1)$. Now, we explain this.

\vspace{0.25em}\noindent\textbf{High Utilization.} As illustrated in Fig.~\ref{Fig.00310}, the function $g_{j}(x)$ is non-increasing; no matter how many self-owned instances a user possesses, a high utilization of them is achieved after
\begin{itemize}
\item we set $\beta_{0}$ to a small enough value in $[0, 1)$.
\end{itemize}
This is because every arriving job will be assigned a large number of self-owned instances when $\beta_{0}$ is small, as illustrated in Fig.~\ref{Fig.0031}.

\vspace{0.25em}\noindent\textbf{Fair Allocation.} Fair allocation means that the allocations of self-owned instances among jobs need to be balanced according to their capabilities of utilizing spot instances. Fair allocation avoids ignoring the difference among jobs and treating them equally where a policy like (\ref{intuitive}) is used; together with Property \ref{proper-1}, the latter can lead to that "rich" jobs (i.e., jobs with strong capabilities where $g_{j}(\beta)$ is small) are consuming unnecessary self-owned instances, i.e.,
\begin{itemize}
\item $r_{j}>g_{j}(\beta)$, where $r_{j}$ denotes the number of self-owned instances allocated to a job; the job's remaining $z_{j}-r_{j}\cdot d_{j}$ workload is expected to be processed by spot instances alone;
\end{itemize}
whereas the others (with large $g_{j}(\beta)$) are allocated poorly and still starving for more self-owned instances, i.e.,
\begin{itemize}
\item $r_{j}<g_{j}(\beta)$; here, on-demand instances are expected to be consumed.
\end{itemize}
Indiscriminate allocations of instances to jobs do harm to the process of achieving the capacity that jobs have for utilizing spot instances, causing unnecessary consumption of more on-demand instances and a higher cost of completing all jobs. In particular, for every rich job, only $g_{j}(\beta)$ self-owned instances are needed to complete its remaining workload without on-demand instances; the saved $r_{j}-g_{j}(\beta)$ self-owned instances can be used for those poorly allocated jobs so as to reduce their consumption of on-demand instances, which improves the cost-efficiency of instance utilization.

Now, we explain that the proposed policy achieves fair allocation by properly setting the value of $\beta_{0}$. The cost-optimal $\beta_{0}$, denoted by $\beta_{0}^{*}$, depends on the statistics of jobs and the amount of self-owned instances available; the online learning technique will be used subsequently in Section~\ref{sec.learning} to infer $\beta_{0}^{*}$. When $\beta_{0}^{*}=0$, self-owned instances themselves are enough to complete all jobs by their deadlines where $g_{j}(\beta_{0})=\frac{z_{j}}{d_{j}}$.

When there are adequate self-owned instances such that $\beta_{0}^{*}\in (0, \beta]$, every arriving job $j$ will be allocated $\geq g_{j}(\beta)$ self-owned instances whenever the amount of idle self-owned is large (i.e., $m_{a_{j}}(a_{j}+d_{j}-1)\geq g_{j}(\beta_{0})$), according to the policy (\ref{policy-long-term}); this is illustrated in Fig.~\ref{Fig.0031} where $\beta=\frac{5}{16}$ and $\beta_{0}^{*}=\frac{1}{16}$. Afterwards, the job $j$ is expected to be completed by utilizing spot instances alone. No job will be allocated $<g_{j}(\beta)$ self-owned instances whenever possible, and fair allocation is achieved. Furthermore, the arriving jobs are allocated on a first come first served basis and we note that $\beta_{0}$ should be properly small but cannot be set to a value too small. If $\beta_{0}$ is too small, jobs that arrive earlier might consume too many self-owned instances and then the jobs that arrive late have less opportunity to get $\geq g_{j}(\beta)$ self-owned instances subject to the availability of these instances (i.e., the value of $m_{a_{j}}(a_{j}+d_{j}-1)$).

When self-owned instances are deficient such that $\beta_{0}^{*}\in (\beta, 1)$, every arriving job will be allocated $< g_{j}(\beta)$ self-owned instances; this is illustrated in Fig.~\ref{Fig.0031} where $\beta=\frac{5}{16}$ and $\beta_{0}^{*}=\frac{31}{64}$. Afterwards, the job $j$ is expected to have to utilizing some amount of on-demand instances to meet its deadline. No job will be allocated $>g_{j}(\beta)$ self-owned instances, achieving fair allocation among jobs: if there exists such allocation, a waste of self-owned instances is caused since we can reduce this allocation to $g_{j}(\beta)$ and allocate these saved instances to other jobs to reduce the consumption of on-demand instances.

\subsection{Spot and On-demand Instances}
\label{sec.spot}

As described in Section~\ref{sec.on-demand-spot}, the instance allocation process is divided into two phases. Now, we analyze the expected optimal strategy to utilize spot instances.

\subsubsection{First phase}
\label{sec.spot-first-phase}

In the first phase, the allocation of $j$ is updated per hour and there is flexibility for $j$ to utilize spot instances. Now, we analyze the expected optimal policy in the first phase. One of the following two cases will happen: (\rmnum{1}) the job $j$ is completed in the first phase, and (\rmnum{2}) in the $i_{j}$-th execution of $j$, after spot instances are terminated by the cloud, there is no flexibility for $j$ to utilize spot instances.

In this paper, the value of $\beta$ is inferred by the online learning technique. If the previous allocation of self-owned instances $r_{j}$ is $\geq g_{j}(\beta)$, it is expected that the first case will happen; then, by Property~\ref{proper-1}, we conclude that
\begin{proposition}\label{proposi-spot-1}
An expected optimal strategy is to bid for $\delta_{j}-r_{j}$ spot instances at every allocation of $j$.
\end{proposition}

\begin{figure*}[t]
	\centering

	\includegraphics[width=6.4in]{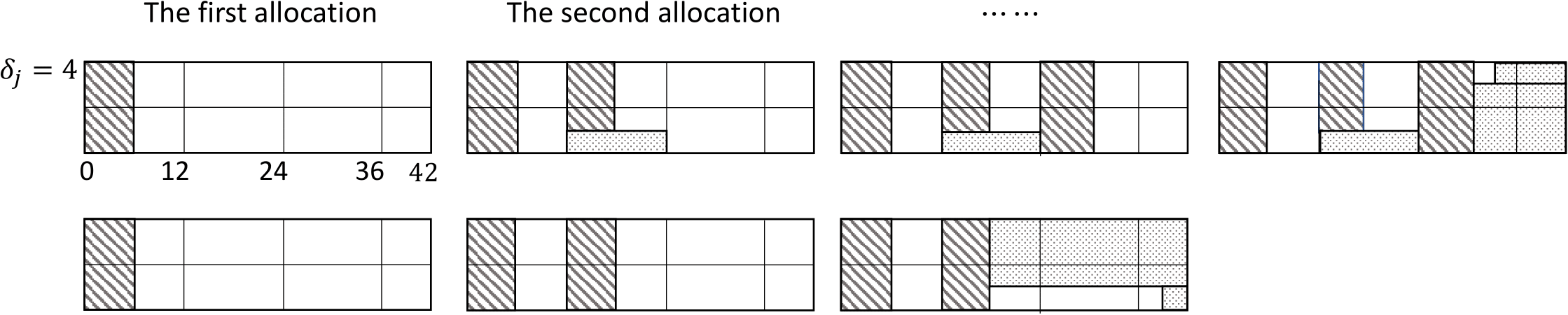}

	\caption{Illustration of Proposition~\ref{proposi-spot} in the case that $\nu(z_{j}, d_{j})<(\kappa_{0}-1)\cdot\delta_{j}$ and $\kappa_{2}(z_{j}, d_{j}) < \kappa_{3}$.}\label{Fig.032}
\end{figure*}

Next, we analyze the optimal strategy when the second case happens. Job $j$ is allocated $r_{j}$ self-owned instances at every $t\in[a_{j}, a_{j}+d_{j}-1]$; afterwards, it can be viewed as a new job with a workload $z_{j}-\delta_{j}\cdot d_{j}$ and a parallelism bound $\delta_{j}-r_{j}$, as described in Section~\ref{sec.self-owned}. So, without loss of generality, we just analyze the optimal strategy in the case where a job $j$ is completed by utilizing on-demand and spot instances alone.

Our decision variables are $o_{j}^{i}$ and $si_{j}^{i}$ where $o_{j}^{i}+si_{j}^{i}=\delta_{j}$. Let $\kappa_{1}$ denote the total number of allocation updates in the first phase where $j$ has flexibility for spot instances; let $\kappa_{0}=\left\lceil d_{j}/Len \right\rceil$ denoting the maximum possible number of allocation updates of $j$ and we have
\begin{align}\label{time-constraint}
\kappa_{1}\leq \kappa_{0}.
\end{align}
In the $i$-th execution of $j$ where $i\in [1, \kappa_{1}]$, it is expected that the workloads processed by spot and on-demand instances are respectively $(\delta_{j}-o_{j}^{i})\cdot Len\cdot \beta$ and $o_{j}^{i}\cdot Len$. By Definition~\ref{def-2}, $j$ has flexibility to utilize unstable spot instances at the $\kappa_{1}$-th allocation, i.e.,
\begin{center}
$s_{j}^{\kappa_{1}} = \frac{\delta_{j}\cdot (d_{j}-(\kappa_{1}-1)\cdot Len)}{z_{j}-\sum\nolimits_{i=1}^{\kappa_{1}-1}{\left( o_{j}^{i}\cdot Len + (\delta_{j}-o_{j}^{i})\cdot Len\cdot\beta\right)}} \geq 1$,
\end{center}
and has no flexibility to utilize spot instances at the $(\kappa_{1}+1)$-th allocation, i.e.,
\begin{center}
$s_{j}^{\kappa_{1}+1} = \frac{\delta_{j}\cdot (d_{j}-\kappa_{1}\cdot Len)}{z_{j}-\sum\nolimits_{i=1}^{\kappa_{1}}{\left(o_{j}^{i}\cdot Len + (\delta_{j}-o_{j}^{i})\cdot Len\cdot\beta\right)}} < 1$.
\end{center}
They are respectively equivalent to the following relations:
\begin{align}
& \sum\nolimits_{i=1}^{\kappa_{1}-1}{(\delta_{j}-o_{j}^{i})\cdot Len\cdot (1-\beta)} \leq d_{j}\cdot\delta_{j}-z_{j},\label{optimal-spot-constraints1}\\
& \sum\nolimits_{i=1}^{\kappa_{1}}{(\delta_{j}-o_{j}^{i})}\cdot Len\cdot (1-\beta) > d_{j}\cdot\delta_{j}-z_{j}. \label{optimal-spot-constraints2}
\end{align}
For the condition that $s_{j}^{\kappa_{1}+1}<1$, a special case is $\kappa_{1}=\kappa_{0}$ where this condition holds trivially since $d_{j}-\kappa_{1}\cdot Len\leq 0$; since $s_{j}^{\kappa_{1}}\geq 1$, the $\kappa_{1}$-th allocation of $j$ is still in the first phase. In this subsection, our objective is to maximize the total workload processed by spot instances at the first $\kappa_{1}$ allocations, i.e.,
\begin{align} \label{optimal-spot}
\text{maximize}\enskip \sum\nolimits_{i=1}^{\kappa_{1}}{(\delta_{j}-o_{j}^{i})\cdot Len\cdot \beta},
\end{align}
subject to the constraints (\ref{time-constraint}), (\ref{optimal-spot-constraints1}), (\ref{optimal-spot-constraints2}), and the constraint that $o_{j}^{i}$ is an integer in $[0, \delta_{j}]$. Our decision variables are $o_{j}^{1}, \cdots, o_{j}^{\kappa_{1}}$.

Now, we give an optimal solution to (\ref{optimal-spot}).
\begin{proposition}\label{proposi-spot-optimal}
An solution to (\ref{optimal-spot}) is optimal if it is of the following form: (\rmnum{1}) $\sum_{i=1}^{\kappa_{1}-1}{(\delta_{j}-o_{j}^{i})}=\min\{ \nu(z_{j}, d_{j}), (\kappa_{0}-1)\cdot\delta_{j}\}$, and (\rmnum{2}) $o_{j}^{\kappa_{1}}=0$,
where
\begin{center}
$\nu(z_{j}, d_{j})  =  \left\lfloor \frac{d_{j} \cdot \delta_{j} - z_{j}}{Len\cdot (1-\beta)} \right\rfloor$.
\end{center}
\end{proposition}

Proposition~\ref{proposi-spot-optimal} indicates the maximum number of spot instances that can be bid for in the first phase, i.e., the maximum value of $\sum\nolimits_{i=1}^{\kappa_{1}}{(\delta_{j}-o_{j}^{i})}$. As a corollary of Proposition~\ref{proposi-spot-optimal}, we conclude that
\begin{proposition}\label{invariable}
Given a job $j$, the expected maximum workload that can be processed by spot instances is
\begin{center}
$\left( \min\left\{ \nu(z_{j}, d_{j}), (\kappa_{0}-1)\cdot\delta_{j}\right\} + \delta_{j} \right)\cdot Len\cdot \beta$.
\end{center}
\end{proposition}

Proposition~\ref{proposi-spot-optimal} also implies an expected optimal strategy for spot instances.

\begin{proposition}	\label{proposi-spot}
	Let $\kappa_{2}(z_{j}, d_{j}) = \lfloor \frac{\nu(z_{j}, d_{j})}{\delta_{j}} \rfloor$ and $\kappa_{3}=\frac{\nu(z_{j}, d_{j})}{\delta_{j}}$. To maximize the total workload processed by spot instances, if $(\kappa_{0}-1)\cdot\delta_{j}\leq \nu(z_{j}, d_{j})$,  we can set $\kappa_{1}=\kappa_{0}$ and {\em an expected optimal strategy} is to
\begin{itemize}
\item bid for $\delta_{j}$ spot instances at each allocation update of $j$.
\end{itemize}
If $\nu(z_{j}, d_{j})< (\kappa_{0}-1)\cdot\delta_{j}$, in the case that $\kappa_{2}(z_{j}, d_{j}) = \kappa_{3}$, we can set $\kappa_{1}=\kappa_{2}(z_{j}, d_{j})+1$ and {\em an expected optimal strategy} is to
\begin{itemize}
\item bid for $\delta_{j}$ spot instances at each of the first $\kappa_{1}$ allocations of $j$, i.e., $o_{j}^{1}=\cdots=o_{j}^{\kappa_{1}}=\delta_{j}$;
\end{itemize}
in the case that $\kappa_{2}(z_{j}, d_{j}) < \kappa_{3}$, we can set $\kappa_{1}=\kappa_{2}(z_{j}, d_{j})+2$ and {\em an expected optimal strategy} is to
	\begin{itemize}
		\item bid for $\delta_{j}$ spot instances at the $1$st, $\cdots$, ($\kappa_{1}-2$)-th, $\kappa_{1}$-th allocations of $j$, i.e., $o_{j}^{1}=\cdots=o_{j}^{\kappa_{1}-2}=o_{j}^{\kappa_{1}}=\delta_{j}$,

		\item bid for $\nu(z_{j}, d_{j}) - \kappa_{2}(z_{j}, d_{j})\cdot \delta_{j}$ spot instances at the ($\kappa_{1}-1$)-th allocation of $j$, i.e., $o_{j}^{\kappa_{1}-1}=\nu(z_{j}, d_{j}) - \kappa_{2}(z_{j}, d_{j})\cdot \delta_{j}$.
	\end{itemize}
\end{proposition}


We illustrate Proposition~\ref{proposi-spot} in Fig.~\ref{Fig.032} where the area of diagonal stripes and the dotted area denote the workload processed respectively by spot and on-demand instances; in the blank area, no workload of $j$ is processed. We assume that $\beta=\frac{1}{2}$ and $L=5$ where $Len=12$; job $j$ has $d_{j}=42$ (3.5 hours), $z_{j}=122$ and $\delta_{j}=4$. Here, we have $\nu(z_{j}, d_{j})=7$ and $\kappa_{2}(z_{j}, d_{j})=1$. From the left to the right, the first four subfigures illustrate the expected optimal allocation. At the first allocation of $j$, $\delta_{j}=4$ spot instances are bid for and the expected execution time of spot instances is $\beta\cdot Len = 6$. At the second allocation of $j$, $(\nu(z_{j}, d_{j})-\delta_{j}\cdot \kappa_{2}(z_{j}, d_{j}))=3$ spot instances are bid for and one on-demand instance is purchased. So far, $\nu(z_{j}, d_{j})=7$ spot instances have been bid for. At the third allocation of $j$, $\delta_{j}$ spot instances are bid for and after the execution of spot instances, $j$ has no flexibility to utilize spot instances and it turns to totally utilize on-demand instances as illustrated by the fourth subfigure. In contrast, we also use the last three subfigures to illustrate {\em an intuitive way} to utilize spot instances where $\delta_{j}$ instances are bid for at every allocation of $j$ when there is flexibility for spot instances. After the execution of spot instances at the second allocation of $j$, it has no flexibility and turns to utilize on-demand instances since $s_{j}^{3}<1$.

As illustrated in Fig.~\ref{Fig.032}, the strategy in Proposition~4.6 can be explained as follows. Whenever possible, bid for the maximum number of spot instances (i.e., $\delta_{j}$ instances). An exception happens only at the second allocation of $j$ where little flexibility is remaining, and we need to properly manage the instance allocation to ensure that there still exists flexibility to utilize spot instances at the third allocation of $j$: then, if $\delta_{j}$ spot instances are bid for at the second allocation, it is expected that there will be no flexibility for $j$ to utilize spot instances at the third allocation; by bidding for less, it could be expected that the allocation will not get into the second phase and there will still be the last flexibility/opportunity at the third allocation of $j$.

Based on Proposition~\ref{proposi-spot}, we propose Algorithm~\ref{proportion} to dynamically determine the numbers of on-demand and spot instances allocated to $j$ at every $i$-th allocation update when there is flexibility to utilize spot instances. At every allocation of $j$ that occurs at slot $t$, the remaining workload of $j$ to be processed could be viewed as a new job with the arrival time $t$, workload $z_{j}^{\prime}$, parallelism bound $\delta_{j}$, and relative deadline $a_{j}+d_{j}-t$; we always use Proposition~\ref{proposi-spot} to determine the first allocation of this new job whose arrival time is $t$.

\begin{algorithm}[!ht]
	\SetKwInOut{Input}{Input}
	\SetKwInOut{Output}{Output}	
	
	\BlankLine	
	
\tcc{\footnotesize{At the $i$-th allocation of $j$, its remaining workload is viewed as a new job with an arrival time $t$, and a relative deadline $a_{j}+d_{j}-t$}}

$\kappa_{0}(t)\leftarrow \left\lceil \frac{a_{j}+d_{j}-t}{Len} \right\rceil$

\tcc{\footnotesize{the case $(\kappa_{0}-1)\cdot\delta_{j}\leq \nu(z_{j}, d_{j})$ in Proposition~\ref{proposi-spot}}}
\If{$(\kappa_{0}(t)-1)\cdot\delta_{j} \leq \nu(z_{j}, a_{j}+d_{j}-t)$}{

	$si_{j}^{i}\leftarrow\delta_{j}$, \enskip $o_{j}^{i}\leftarrow 0$\;
}
\Else{	

\tcc{\footnotesize{both cases $\kappa_{2}(z_{j}, d_{j}) = \kappa_{3}$ and $\kappa_{2}(z_{j}, d_{j}) < \kappa_{3}$ where $\kappa_{2}(z_{j}, d_{j})\geq 1$}}
	\If{$\kappa_{2}(z_{j}^{\prime}, a_{j}+d_{j}-t)\geq 1$}{
		
		$si_{j}^{i}\leftarrow\delta_{j}$, \enskip $o_{j}^{i}\leftarrow 0$\;
		
	}

\tcc{\footnotesize{the case $\kappa_{2}(z_{j}, d_{j}) < \kappa_{3}$ where $\kappa_{2}(z_{j}, d_{j}) = 0$}}
	\If{$\kappa_{2}(z_{j}^{\prime}, a_{j}+d_{j}-t) = 0$ $\wedge$ $\nu(z_{j}, a_{j}+d_{j}-t)>0$}{
		
		$si_{j}^{i}\leftarrow \nu(z_{j}, a_{j}+d_{j}-t)$, \enskip $o_{j}^{i}\leftarrow \delta_{j} - si_{j}^{i}$\;
		
	}

\tcc{\footnotesize{the case $\kappa_{2}(z_{j}, d_{j}) = \kappa_{3}=0$}}
	\If{$\nu(z_{j}, a_{j}+d_{j}-t)=0$}{
			
	    $si_{j}^{i}\leftarrow \delta_{j}$, \enskip $o_{j}^{i}\leftarrow 0$\;
	
    }	
	
}
	
	$b_{j}^{i}\leftarrow b$\;
	
	at the $i$-th allocation, bid a price $b_{j}^{i}$ for $si_{j}^{i}$ spot instances\;

	\caption{Proportion($j$, $\beta$, $b$)\label{proportion}}
\end{algorithm}

\subsubsection{Second phase}
\label{sec.spot-second-phase}

As described in Section~\ref{sec.on-demand-spot}, once spot instances get lost at every allocation of $j$, the scheduler uses Definition~\ref{def-2} to check the flexibility to utilize spot instances. In the $i_{j}$-th execution, when spot instances get lost at the beginning of some slot $t_{1}^{\prime}$, there is no such flexibility; then, the instance allocation enters the second phase where only on-demand instances are utilized. Now, we analyze their optimal utilization.

\begin{figure}[!ht]
	\centering
	\includegraphics[width=3.1in]{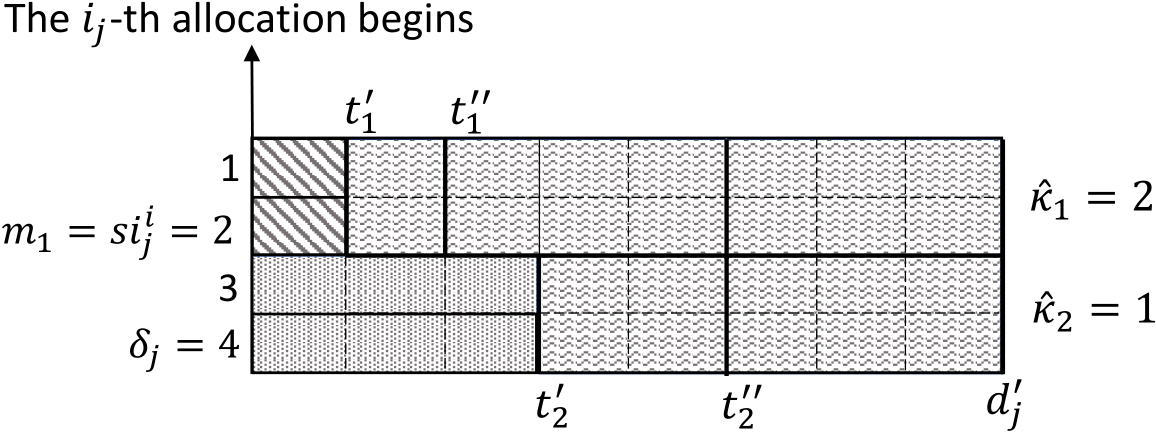}\\
	\caption{The second phase of allocation where $i=i_{j}$: the area of waves denotes the available space in the second phase; the area of diagonal stripes and the dotted area respectively denote the workload processed in the $i_{j}$-th execution by spot instances and on-demand instances that are utilized for an hour.}\label{Fig.422}
\end{figure}

As shown in Algorithm~\ref{Slicing}, at every allocation of $j$ in the first phase (including the $i_{j}$-th allocation), the number of on-demand instances allocated to $j$ is either 0 (see lines 3, 6, 10) or $>0$ (see line 8). Let $t_{2}^{\prime}=a_{j}+i_{j}\cdot Len$, $d_{j}^{\prime}=a_{j}+d_{j}-1$, and we define two parameters that represent the maximum multiple of an hour (containing $Len$ slots) respectively in time intervals $[t_{1}^{\prime}, d_{j}^{\prime}]$ and $[t_{2}^{\prime}, d_{j}^{\prime}]$:
\begin{center}
$\hat{\kappa}_{1}=\left\lfloor \frac{d_{j}^{\prime}-(t_{1}^{\prime}-1)}{Len} \right\rfloor$, and $\hat{\kappa}_{2}=\left\lfloor \frac{d_{j}^{\prime}-t_{2}^{\prime}+1}{Len} \right\rfloor$;
\end{center}
Let $t_{i}^{\prime\prime}=d_{j}^{\prime}-\hat{\kappa}_{i}\cdot Len+1$ ($i\in \{1, 2\}$), and after deducting $\hat{\kappa}_{1}$ and $\hat{\kappa}_{2}$ hours respectively from the two intervals, the numbers of remaining slots in $[t_{1}^{\prime}, t_{1}^{\prime\prime}-1]$ and $[t_{2}^{\prime}, t_{2}^{\prime\prime}-1]$ are denoted by $\phi_{1}$ and $\phi_{2}$:
\begin{center}
$\phi_{1} =t_{1}^{\prime\prime} - t_{1}^{\prime}$,\enskip and\enskip $\phi_{2} = t_{2}^{\prime\prime} - t_{2}^{\prime}$,
\end{center}
where $0\leq \phi_{1}, \phi_{2}<Len$. The related notation is also illustrated in Fig.~\ref{Fig.422}. Let
\begin{center}
$m_{0}=si_{j}^{i}\cdot\hat{\kappa}_{1}+o_{j}^{i}\cdot\hat{\kappa}_{2}$,\,\, $m_{1}=si_{j}^{i}$,\,\, and $m_{2}=o_{j}^{i}$,
\end{center}
where $i=i_{j}$. In Fig.~\ref{Fig.422}, the available space in the second phase is the area of waves and $m_{0}$ represents the maximum integer of instance hour that can be utilized by $j$.

Since every on-demand instance is charged on an hourly basis, a cost-optimal strategy in the second phase is to minimize the integer instance hours (i.e., the number of on-demand instances $\times$ the time for which they are utilized). The following conclusion possibly is intuitive although a formal proof is also provided: whenever an instance is purchased for an hour, it should be utilized as long as possible with the space constraint.
\begin{proposition}\label{opti-1}
Let $y=y_{0}+y_{1}+y_{2}$ be the minimum such that $y_{0}\cdot Len + y_{1}\cdot\phi_{1} + y_{2}\cdot\phi_{2} \geq z_{j}^{i_{j}+1}$ subject to $y_{0}$, $y_{1}$, $y_{2}$ are non-negative integers and $y_{0}\in [0, m_{0}]$,\, $y_{1}\in [0, m_{1}]$,\,\, $y_{2}\in [0, m_{2}]$. In the second phase, a cost-optimal strategy is to purchase on-demand instances for $y$ instance hours\footnote{The specific value of $y$ is given while proving this proposition, which can be found in the appendix.}.
\end{proposition}

\subsection{Scheduling Framework}
\label{sec.framework}

As described above, a general policy is defined by a tuple $\{\beta_{0}, \beta, b\}$ and determines the amounts of self-owned, spot, and on-demand instances allocated to a job, and the bid price. The instance allocation process has been described in Section~\ref{sec.rules}. Based on this, at every slot $t$, if a job $j$ just arrives or it has arrived before but not been completed yet, we propose a framework, presented in Algorithm~\ref{Slicing}, to determine the action of allocating instances to $j$ after checking the state of $j$. Actions are needed in the following three states: (\rmnum{1}) $t$ is the arrival time of $j$, determining the allocation of self-owned instances, (\rmnum{2}), $t$ equals $a_{j}+ (i-1)\cdot Len$ where the $i$-th allocation update of spot and on-demand instances needs to be done, (\rmnum{3}) the spot instances of $j$ get lost at $t$ where we need to check whether $j$ still has flexibility for spot instances. In Algorithm~\ref{Slicing}, $z_{j}^{\prime}$ denotes the remaining workload of $j$ to be processed after deducting its current allocations from $z_{j}$; upon arrival of $j$, $z_{j}^{\prime}=z_{j}$.

\subsection{The Application of Online Learning}
\label{sec.learning}

Whenever a job $j$ arrives, the instance allocation process at every slot $t\in [a_{j}, a_{j}+d_{j}-1]$ is regulated by Algorithm~\ref{Slicing} and the specific amounts of various instances allocated to $j$ is determined by $\beta_{0}, \beta, b$. To learn the most cost-effective parameters, we apply the online learning algorithm (TOLA) in \cite{Jain14}. We present here its main idea; a formal description can be found in the appendix.

There are a set of jobs $\mathcal{J}$ that arrive over time, indexed by $j=1, 2, \cdots$, and a set of $n$ parametric policies $\mathcal{P}$ each specified by $\{\beta_{0}, \beta, b\}$ and indexed by $\pi=1, 2, \cdots$. Let $d=\max_{j\in\mathcal{J}}\{d_{j}\}$, i.e., the maximum relative deadline of all jobs, and $\mathcal{J}_{t}\subseteq \mathcal{J}$ denote all jobs $j$ arriving at slot $t$, i.e., $a_{j}=t$. There is a weight distribution $w$ over $n$ policies whose initial value is $\{1/n, \cdots, 1/n\}$.
Time goes from slot $1$ towards later slots. At every slot $t$, TOLA randomly picks for a job $j\in\mathcal{J}_{t}$ a policy $\pi_{j}$ from $\mathcal{P}$ according to the current $w$ and bases the allocation of instances to $j$ on that policy. In the meantime, the distribution $w$ will also be updated at every $t>d$. At such $t$, we have the knowledge of spot prices in $[t-d, t-1]$ and can derive the cost of completing a job $j^{\prime}\in\mathcal{J}_{t-d}$ under every policy $\pi\in\mathcal{P}$; the distribution is updated such that the lower-cost (higher-cost) polices of this job are re-assigned the enlarged (resp. reduced) weights. Thus, as time goes by and more and more jobs are processed, the most cost-effective policies of $\mathcal{P}$ will be identified gradually, i.e., the ones with the highest weights. When $t$ is large, TOLA will choose the most cost-effective policy for every arriving job and the actual cost of completing all jobs is close to the cost of completing all jobs under a specific policy $\pi^{*}\in\mathcal{P}$ that generates the lowest total cost.

\begin{algorithm}[!ht]
\SetKwInOut{Input}{Input}
\SetKwInOut{Output}{Output}

\Input{the job's current characteristics $\{a_{j}, d_{j}, z_{j}^{\prime}, \delta_{j}\}$ where $z_{j}^{\prime}$ is still $>0$, and a parameterized policy $\{\beta_{0}, \beta, b\}$}

\tcc{\footnotesize{allocate instances at the very beginning of slot $t$}}

    \If{$a_{j}=t$}{
		
        \tcp{\footnotesize{upon arrival of $j$, allocate self-owned instances to it}}
		
         set the value of $r_{j}$ using Equation~(\ref{policy-long-term})\;

        \For{$\overline{t}\leftarrow a_{j}$ \KwTo $a_{j}+d_{j}-1$}{
           $r_{j}(\overline{t}) \leftarrow r_{j}$\;
        } 		    				
	}

  $i\leftarrow \left\lfloor  \frac{t-a_{j}}{Len} \right\rfloor + 1$\tcp{\footnotesize{used to number the allocation update}}

  \vspace{0.15em}
  \If{$\frac{t-a_{j}}{Len} = i-1$}{

  \tcp{\footnotesize{at the $i$-th allocation of $j$ where it has flexibility for spot instances}}

        \If{$r_{j}\geq g_{j}(\beta)$}{

           \tcp{\footnotesize{it is expected that $j$ will be completed by utilizing spot instances alone after allocating self-owned instances}}

           apply the strategy in Proposition~\ref{proposi-spot-1} here\;

        }
        \Else{

           call Algorithm~\ref{proportion}\;

        }

  }

     \If{the spot instances of $j$ get lost at the beginning of slot $t$}{

     \If{$\frac{(\delta_{j}-r_{j})\cdot (d_{j}-Len\cdot i)}{z_{j}^{\prime}} < 1 $}{

        \tcp{\footnotesize{$j$ has no flexibility to utilize spot instances at the next allocation update by Definition~\ref{def-2}}}

        apply the strategy in Proposition~\ref{opti-1} here\;

     } \footnotesize{\tcp{otherwise, $j$ still has the flexibility at the next allocation update where $z_{j}^{\prime}=z_{j}^{i+1}$}}

}
\caption{Dynalloc($a_{j}, d_{j}, z_{j}^{\prime}, \delta_{j}, \beta_{0}, \beta, b, N, t$)}\label{Slicing}
\end{algorithm}

\subsection{Extension to Microsoft Azure Cloud}


Above, we are essentially studying the following question. On-demand instances are always available and charged a fixed unit price. The availability of spot instances is uncertain over time; intuitively, it is the probability that a user successfully gets spot instances and we denote by $\beta$ its average value. There is a fixed number of self-owned instances. The costs of utilizing self-owned, spot and on-demand instances are increasing. Our question is about the cost-effective strategy to utilize these instances. Our intuition is to maximize the utilization of self-owned instances; when they are not adequate for completing a job, we aim to minimize the utilization of costly on-demand while maximizing the utilization of spot instances. The availability $\beta$ and the job characteristics (deadline, workload, parallelism bound) determine the unique capability of each job $j$ to utilize spot instances, i.e., the maximum workload that could be processed by spot instances. We also give the minimum amount $r_{j}^{min}$ of self-owned instances that each job $j$ needs to complete itself without utilizing any costly on-demand instances. Based on this, related policies are proposed to allocate instances to jobs.

So far, this question has been addressed in the context of Amazon EC2 pricing. On-demand instances are charged on an hourly basis. The price of spot instances (i.e., spot price) fluctuates over time and is updated every 5 minutes. Every user bids a price for spot instances and only if its bid price is not below the current spot price, it could utilize the spot instances for at least 5 minutes. Spot users are charged according to the spot prices. Under such context, the availability of spot instances depends on the bid price of users and the spot prices.

Beyond Amazon EC2, Microsoft Azure began to offer low-priority VMs (virtual machines) since May 2017, as well as high-priority VMs \cite{Microsoft}; it is the second largest IaaS service provider and accounts for 13.3\% of the global market share in 2017 \cite{Gartner}. High- and low-priority VMs respectively correspond to on-demand and spot instances only with some difference in pricing. In Microsoft Azure, high-priority VMs are always available and charged a fixed price; also, low-priority VMs have a lower price but their availability varies over time. However, its pricing model simplifies the Amazon EC2 pricing in that (\rmnum{1}) high-priority VMs are charged per second instead of on an hourly basis, and (\rmnum{2}) the price of low-priority VMs is fixed but their availability is a system-level random variable, without depending on the bid price of users. So, we can roughly say that users will be billed for the exact period when VMs are utilized, without rounding up partial instance hour to full hour.

Now, we explain how to apply the framework of this paper to the Microsoft Azure scenario. Upon arrival of a job $j$, we can still use the policy (\ref{policy-long-term}) proposed in Section~\ref{sec.self-owned} for the allocation of self-owned instances: in the case that they are not adequate, the number of allocated instances is no larger than $r_{j}^{min}$; in the opposite case, the number is no smaller than $r_{j}^{min}$. After allocating self-owned instances, the job $j$ can be viewed as a new job with a reduced parallelism bound $\delta_{j}-r_{j}$. Unlike the case in Amazon EC2, it is not necessary to update the allocation of on-demand instances and bid a price for spot instances every hour. From its arrival on, job $j$ continuously attempts to utilize $\delta_{j}-r_{j}$ low-priority VMs; once there is no flexibility for $j$ to utilize such VMs at some moment, $j$ turns to utilize $\delta_{j}-r_{j}$ high-priority VMs until it is completed. In the allocation process above, only one parameter $\beta_{0}$ is needed to control the allocation of self-owned instances; in contrast, there are three parameters $\{\beta_{0}, \beta, b\}$ in the case of Amazon EC2. So, when the approach of online learning is applied here, only $\beta_{0}$ is needed to be learned.

\section{Evaluation}
\label{sec.evaluation}

The main aim of our evaluations is to show the effectiveness of the proposed policies of this paper.

\subsection{Simulation Setups}
\label{sec.parameters}

The on-demand price is $p=0.25$ per hour. We set $L$ to 5 (minutes) and all jobs have a parallelism bound of 20. Following \cite{Chen11a,Zheng16}, we generate the jobs as follows. The job's arrival is generated according to a poisson distribution with a mean of 2. The size $z_{j}$ of every job $j$ is set to $12\times 20\times x$ where $x$ follows a bounded Pareto distribution with a shape parameter $\epsilon = \frac{1}{1.01}$, a scale parameter $\sigma=\frac{1}{6.06}$ and a location parameter $\mu=\frac{1}{6}$; the maximum and minimum value of $x$ is set to 1 and 10. The job's relative deadline is generated as $x\cdot z_{j}/\delta_{j}$, where $x$ is uniformly distributed over $[1, x_{0}]$. $x$ represents the slackness of a job; it affects the jobs' capability to utilize spot instances as shown by Proposition~\ref{invariable}, and is a main factor that determines the performance. In this paper, we consider three types of jobs respectively with a small, medium, and large slackness: the 1st, 2nd, 3rd types of jobs respectively with $x_{0}=3, 7, 13$. Spot prices are updated every time slot and their values can follow an exponential distribution where its mean is set to 0.11 \cite{Zheng15}.

\vspace{0.15em}\noindent\textbf{Proposed Policies.} The policies of this paper are parameterized: $\beta$ and $b$ are used for determine the allocation of spot and on-demand instances (see lines 5-13 of Algorithm~\ref{Slicing}), and $\beta_{0}$ is for self-owned instances (see lines 1-4 of Algorithm~\ref{Slicing}). The parameter $\beta_{0}$ is chosen in $\mathcal{C}_{1}=\{\frac{i}{10} \mid 0\leq i\leq 6\}$. As illustrated in Fig.~\ref{Fig.00310}, for jobs with $x_{0}>1.25$, the amount of self-owned instances allocated to jobs can be effectively controlled by selecting a value $\leq 0.6$; for the others with little flexibility to utilize spot instances, they will be a large number of self-owned instances whenever possible to reduce the consumption of on-demand instances. The parameter $\beta$ is chosen from $\mathcal{C}_{2} = \{\frac{i}{10} \mid 0\leq i\leq 9 \}\cup\{0.9999\}$. The bid price $b$ is chosen in $\mathcal{B}=\{ b_{i}=0.13 + 0.03\cdot(i-1) \mid 1\leq i\leq 6\}$. When only spot and on-demand instances are considered, let $\boldsymbol{\mathcal{P}}=\{(\beta, b) \mid \beta\in \mathcal{C}_{2}, b\in\mathcal{B} \}$, representing all policies of this paper to be evaluated; when self-owned instances are also taken into account, let $\boldsymbol{\mathcal{P}}=\{(\beta, b, \beta_{0}) \mid \beta_{0}\in\mathcal{C}_{1}, \beta\in \mathcal{C}_{2}, b\in\mathcal{B} \}$.

\vspace{0.15em}\noindent\textbf{Compared Policies.} The policies of this paper are compared with (\rmnum{1}) the naive policy (\ref{intuitive}) for self-owned instances and (\rmnum{2}) the policy proposed in \cite{Jain14} only for spot and on-demand instances (see Algorithm~1 in \cite{Jain14}). The latter randomly selects a parameter $\theta\in\Theta=\{ \frac{i}{10} \mid 0\leq i\leq 10 \}$ for every job $j$: (\rmnum{1}) the user will bid a price $b$ for $\theta\cdot \delta_{j}$ spot instances and acquire $(1-\theta)\cdot \delta_{j}$ on-demand instances at every allocation update of $j$; (\rmnum{2}) it monitors at every slot $t$ whether there is a risk of not completing the job by its deadline if only $(1-\theta)\cdot \delta_{j}$ on-demand instances are utilized in the remaining slots; (\rmnum{3}) if such risk exists, there is no flexibility for utilizing spot instances and it turns to utilize $\min\left\{\delta_{j}, \left\lceil z_{j}^{i_{j}+1}/Len \right\rceil\right\}$ on-demand instances alone until $j$ is completed\footnote{In \cite{Jain14}, the workload of $j$ is measured in instance hours.}. Let $\boldsymbol{\mathcal{P}^{\prime}} = \{(\theta, b) \mid \theta\in \Theta, b\in\mathcal{B} \}$, representing all the policies of \cite{Jain14}.

\vspace{0.15em}\noindent\textbf{Performance Metric.} Let $\pi$ denote a policy in $\mathcal{P}$ or $\mathcal{P}^{\prime}$. Given a set of jobs $\mathcal{J}$ that arrive over time, our aim is to minimize the cost of completing all jobs in $\mathcal{J}$; and a main performance metric is {\em the average unit cost of processing jobs} when the $x_{2}$-th type of jobs are processed with $x_{1}$ self-owned instances available, i.e.,
\begin{itemize}
\item the ratio of the total cost of utilizing various instances to the processed workload of jobs, denoted by $\alpha_{x_{1}, x_{2}}$, where $\alpha_{x_{1}, x_{2}} = \sum_{j\in\mathcal{J}}{c_{j}(\pi)}/\sum_{j\in\mathcal{J}}{z_{j}}$.
\end{itemize}
When a policy in $\mathcal{P}$ or $\mathcal{P}^{\prime}$ is applied to process all jobs, we denote by $\alpha_{x_{1}, x_{2}}(\pi)$ the corresponding average unit cost of processing jobs. Against the unknown statistics of spot prices and job's characteristics, there are some policies in $\mathcal{P}$ or $\mathcal{P}^{\prime}$ that are the most cost-effective. We use $\alpha_{x_{1}, x_{2}}$ (resp. $\alpha_{x_{1}, x_{2}}^{\prime}$) to denote the minimum of the average unit costs of our policies (resp. the policies in \cite{Jain14} and defined by (\ref{intuitive})), where $x_{2}=1, 2$, e.g., $\alpha_{x_{1}, x_{2}}=\min\nolimits_{\pi\in\mathcal{P}}\{\alpha_{x_{1}, x_{2}}(\pi)\}$.

The performance of the intuitive policy (\ref{intuitive}) (for self-owned instances) and the existing policy in \cite{Jain14} (for spot and on-demand instances) are used as {\em the baseline} to measure the performance of the proposed policies; so, one performance indicator can be as follows:
\begin{center}
$\rho_{x_{1}, x_{2}} = 1 - \frac{\alpha_{x_{1}, x_{2}}}{\alpha_{x_{1}, x_{2}}^{\prime}}$;
\end{center}
it represents the performance improvement of the proposed policies $\mathcal{P}$ over the baseline, that is, the ratio in cost reduction. Moreover, in this paper, the online learning algorithm TOLA is run to actually select a policy for each arriving job. The selection is random according to a distribution that will be updated according to the cost of completing that job; after numerous jobs are processed, the policies that generate the lowest cost will be associated with the highest probability. In this case, we use $\overline{\alpha}_{x_{1}, x_{2}}(\mathcal{P})$ or $\overline{\alpha}_{x_{1}, x_{2}}(\mathcal{P}^{\prime})$ to denote the average unit cost of processing jobs when $\mathcal{P}$ or $\mathcal{P}^{\prime}$ is applied to TOLA. When online learning is applied, the performance indicator can be as follows:
\begin{center}
$\overline{\rho}_{x_{1}, x_{2}} = 1-\frac{\overline{\alpha}_{x_{1}, x_{2}}(\mathcal{P})}{\overline{\alpha}_{x_{1}, x_{2}}(\mathcal{P}^{\prime})}$;
\end{center}
it represents the ratio in cost reduction when online learning is applied.

\subsection{Results}

In the following, we give the results of simulations that are taken over about 60000 jobs, mainly listed in Tables~\ref{table-spot},~\ref{table-self-owned},~\ref{table-hybrid-use}, and~\ref{table-regret}. In our simulations, all fractional solutions will be rounded up to the nearest integers.

\vspace{0.25em}\noindent\textbf{Experiment 1.} We aim to evaluate the effectiveness of the proposed policies $\mathcal{P}$ for spot and on-demand instances alone by means of comparisons with the policies $\mathcal{P}^{\prime}$ in \cite{Jain14}, where $x_{1}=0$. The simulation results are listed in Table~\ref{table-spot} and show a noticeable cost reduction by up to 64.51\%.

\begin{table}[!ht]
	\centering
		\caption{Performance Improvements for Spot and On-Demand Instances}
	\begin{threeparttable}[b]

		\begin{tabular}{|C{2.4cm}| C{2.4cm} | C{2.4cm} |}   
			
			\hline
			                $\rho_{0, 1}$   & $\rho_{0, 2}$ & $\rho_{0,3}$  \\ \hline			
			                  58.87\%       &     60.84\%   &  64.51\%     \\ \hline	
		\end{tabular}
		\label{table-spot}
	\end{threeparttable}
\end{table}

There are a total of 66 policies in $\mathcal{P}$. In our simulations, every 11 policies are grouped together and they use the same bid price. We have in the same group of policies that the cost-optimal value of $\beta$ (denoted by $\beta^{*}$) is the same even under different types of jobs; the particular results are illustrated in Table~\ref{table-beta-value}. So, in the rest of our simulations, the effective range of $\beta$ will be defined in $\{0.5, 0.6, 0.7, 0.8, 0.9, 0.999999\}$, to which we reset the value of $\mathcal{C}_{2}$.

\begin{table}[!ht]
	\centering
		\caption{The Optimal $\beta$ under a Bid Price $b$ }
	\begin{threeparttable}[b]

		\begin{tabular}{|C{0.5cm}|C{0.6cm}|C{0.6cm}|C{0.6cm}|C{0.6cm}|C{1.05cm}|C{1.05cm}|}   
			
			\hline
			            $b  $   & 0.13   & 0.16  & 0.19 & 0.22 & 0.25 & 0.28  \\ \hline			
			           $\beta$    &  0.7  &  0.8  &  0.9    &  0.9  &  0.999999 & 0.999999     \\ \hline	
		\end{tabular}
		\label{table-beta-value}
	\end{threeparttable}
\end{table}

\vspace{0.25em}\noindent\textbf{Experiment 2.} We aim to evaluate the proposed policy for self-owned instances, compared with the naive policy in (\ref{intuitive}); here, the allocation of spot and on-demand instances will use the same policy $\mathcal{P}$ proposed in this paper. The simulation results are listed in Table~\ref{table-self-owned}, showing a noticeable cost reduction by up to 43.74\%.

\begin{table}[!ht]
	\centering
		\caption{Performance Improvement for Self-Owned Instances}
	\begin{threeparttable}[b]
		\begin{tabular}{|C{0.9cm}| C{1.35cm} |  C{1.35cm} |  C{1.35cm} |  C{1.35cm} |}
			\hline
			       &      $\rho_{200, x_{2}}$   & $\rho_{400, x_{2}}$   & $\rho_{600, x_{2}}$   & $\rho_{800, x_{2}}$  \\ \hline
       $x_{2}=1$   &   15.73\%  &  21.41\%  &  27.07\%  &  22.83\%   \\ \hline	
       $x_{2}=2$   &   27.25\%   &  39.59\%   &   34.04\%   &   17.85\%   \\ \hline
       $x_{2}=3$   &   33.05\%  &   34.41\%   &  43.74\%  &   31.88\%  \\ \hline
		\end{tabular}
		\label{table-self-owned}
	\end{threeparttable}
\end{table}

The utilizations of self-owned instances under different policies are illustrated in Fig.~\ref{Fig.3a}, where the dotted lines from top to down respectively represents the case where $x_{1}=200, 400, 600$ and 800; the particular results are given by the stars on the same dotted line. The allocation of self-owned instances are determined by the policy (\ref{policy-long-term}) or (\ref{intuitive}). Given a set of jobs, the utilization of self-owned instances under the policy (\ref{policy-long-term}) only depends on the parameter $\beta_{0}$ since their allocation is before and independent of the allocation of spot and on-demand instances. The intuitive policy (\ref{intuitive}) is a special form of the policy (\ref{policy-long-term}) when $\beta_{0}=0$. In the case that $x_{2}=2$, when $x_{1}=200, 400, 600, 800$, the minimum average unit cost is generated when $\beta=0.3$, 0.2, 0.2, 0.1 respectively; the corresponding utilizations are given in Table~\ref{table-instance-utilization-1}; the utilization of the intuitive policy (\ref{intuitive}) is illustrated in Table~\ref{table-instance-utilization-2}. We can see that, given a case of $x_{1}$ and $x_{2}$, the proposed policy achieves a lower utilization than the intuitive policy; even so, it still achieves a lower average unit cost as shown in Table~\ref{table-self-owned} where $x_{2}=2$. This is because the proposed policy could effectively reduce the unnecessary consumption of on-demand instances as explained in Section~\ref{sec.explanation}.

\begin{table}[!ht]
	\centering
	\begin{threeparttable}[b]
		\caption{The Instance Utilization of the Proposed Policy under Cost-Optimal $\beta_{0}$}
		\begin{tabular}{| C{1.35cm} | C{1.15cm} |  C{1.15cm} |  C{1.15cm} |  C{1.15cm} |}
			\hline
     $(\beta_{0}, x_{1})$  &   (0.3, 200)   & (0.2, 400)  &   (0.2, 600)  & (0.1, 800)  \\ \hline

      Utilization &  89.89\%   &  92.41\%   &   72.70\%   &   96.39\%   \\ \hline
       		\end{tabular}
		\label{table-instance-utilization-1}
	\end{threeparttable}
	\begin{threeparttable}[b]
		\caption{The Instance Utilization of the Intuitive Policy}
		\begin{tabular}{| C{1.55cm} | C{1.1cm} |  C{1.1cm} |  C{1.1cm} |  C{1cm} |}
			\hline
     $x_{1}$  &   200   &  400  &  600  &  800  \\ \hline

     Utilization  &  99.73\%  &  99.57\%  &  99.31\%  &  98.89\%   \\ \hline

       		\end{tabular}
		\label{table-instance-utilization-2}
	\end{threeparttable}
\end{table}

\begin{figure}[!ht]
  \centering
	\includegraphics[width=3.30in]{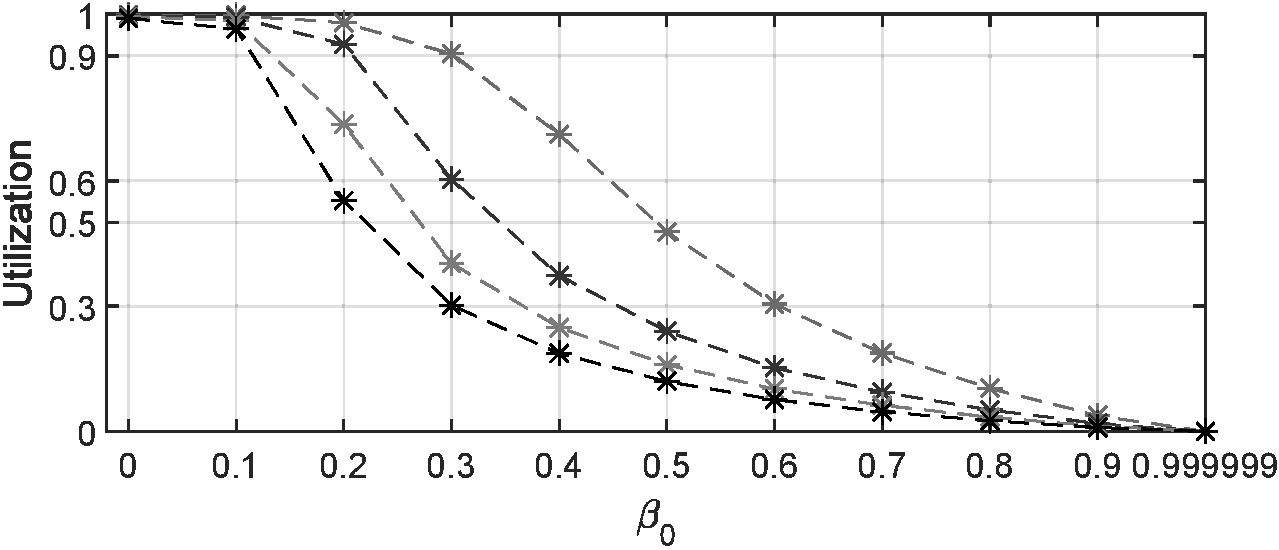}
	\caption{The utilization of self-owned instances under different values of $\beta_{0}$.}
    \label{Fig.3a}
\end{figure}

\vspace{0.25em}\noindent\textbf{Experiment 3.} Assume that there are some amount of self-owned instances, and we show the performance improvement of the proposed policies $\mathcal{P}$, compared with the policies that use $\mathcal{P}^{\prime}$ for spot and on-demand instances and (\ref{intuitive}) for self-owned instances. The simulation is done under the 2nd type of jobs that have a medium slackness, and the results are listed in Table~\ref{table-hybrid-use}, showing the improvement of performance by up to 75.68\%.

\begin{table}[!ht]
	\centering
		\caption{Performance Improvement for Three Types of Instances}
	\begin{threeparttable}[b]
		\begin{tabular}{| C{1.55cm} |  C{1.55cm} |  C{1.55cm} |  C{1.55cm} |}
			\hline
     $\rho_{200, 2}$   & $\rho_{400, 2}$   & $\rho_{600, 2}$   & $\rho_{800, 2}$  \\ \hline

        71.30\%   &  75.68\%   &   72.83\%   &   66.65\%   \\ \hline

       		\end{tabular}
		\label{table-hybrid-use}
	\end{threeparttable}
\end{table}

\vspace{0.25em}\noindent\textbf{Experiment 4.} Now, we show the performance of the proposed policies when online learning is applied. The simulation setting is the same as Experiment 3. The related results are illustrated in Table~\ref{table-regret}, showing a cost reduction by up to 66.71\%.

\begin{table}[!ht]
	\centering
		\caption{Performance Improvement under Online Learning}
	\begin{threeparttable}[b]
		\begin{tabular}{| C{1.2cm} | C{1.2cm} |  C{1.2cm} |  C{1.2cm} |  C{1.2cm} |}
			\hline
		$\overline{\rho}_{0,2}$    &  $\overline{\rho}_{200,2}$   & $\overline{\rho}_{400,2}$   & $\overline{\rho}_{600, 2}$   & $\overline{\rho}_{800, 2}$  \\ \hline
        60.89\%       &  63.28\%  &  66.71\%  &  63.60\%  &  51.11\%   \\ \hline	

		\end{tabular}
		\label{table-regret}
	\end{threeparttable}
\end{table}

\section{Concluding Remark}
\label{sec.conclusion}

Utilizing IaaS clouds cost-effectively is an important concern for all users. In this paper, we consider the problem of how to utilize different purchase options including spot and on-demand instances, in addition to possibly existing self-owned instances, to minimize the cost of processing all incoming jobs while respecting their response-time targets. Driven by the goal of maximizing the utilization of self-owned instances while optimizing the possibility of utilizing spot instances, we answer two underlying questions in the instance allocation process: to be cost-effective, what properties should be kept in the policy for allocating self-owned instances and what policy can maximize the utilization of spot instances, escaping unnecessary consumption of costly on-demand instances.

As a result, we propose parametric policies for the allocation of these three types of instances that achieve small costs. The proposed policies are adaptive and, facing the dynamic of cloud market, these policies use online learning to infer the optimal values of their parameters. Through numerical simulations, we show the effectiveness of our proposed policies, in particular that they achieve a cost reduction of up to 64.51\% when spot and on-demand instances are considered and of up to 43.74\% when self-owned instances are considered. In future, we will extend the framework of this paper to process precedence-constrained jobs.



\ifCLASSOPTIONcompsoc
  \section*{Acknowledgments}
\else
  \section*{Acknowledgment}
\fi

The work of Patrick Loiseau was supported by the French National Research Agency (ANR) through the Investissements davenir program (ANR-15-IDEX- 02), and by the Alexander von Humboldt Foundation. Part of Xiaohu Wu's work was done when he was with Eurecom, Sophia-Antipolis, France; in addition, his work was also supported by the European Union's Horizon 2020 research and innovation programme in the ROMA project (grant no.\ 754514). The work of Esa Hyyti\"a was supported by the Academy of Finland in the FQ4BD project (grant no.\ 296206).

\appendices

\section{Proofs of Propositions}

This section contains the proofs of the propositions in the Section 4.

\vspace{0.7em}\noindent\textbf{Proof of Proposition 4.1.} Assume that a job $j$ is allocated $r_{j}$ self-owned instances in $[a_{j}, a_{j}+d_{j}-1]$. At each of the first $\kappa_{0}$ allocations of $j$, the expected time of utilizing spot instances is $\beta\cdot Len$. If a job can be expected to be completed by the deadline by totally utilizing spot instances after the allocation of self-owned instances, we have that (\rmnum{1}) it could be expected that the workload processed by self-owned instances plus the workload processed by spot instances at every allocation of $j$ is no less than $z_{j}$, and (\rmnum{2}) after the allocation of self-owned instances, the allocation of spot and on-demand instances is always in the first phase as described in the Section~\ref{sec.rules}, i.e., the allocation is updated every hour where only spot instances are bid for.

Now, we analyze two cases. The first one is $d_{j}-\kappa_{0}\cdot Len > \beta\cdot Len$. In this case, in the ($\kappa_{0}+1$)-th execution of $j$, the expected time of utilizing spot instances is $\beta\cdot Len$; then, it is expected that
   \begin{align*}
       r_{j}\cdot d_{j} + (\kappa_{0}+1)\cdot (\delta_{j}-r_{j})\cdot Len\cdot \beta \geq z_{j}.
   \end{align*}
   This leads to that $r_{j} \geq r_{j}^{\prime}(\beta)$. The second case is $d_{j}-\kappa_{0}\cdot Len \leq \beta\cdot Len$. In this case, in the ($\kappa_{0}+1$)-th execution of $j$, the expected time of utilizing spot instances is $\min\{\beta\cdot Len, d_{j}-\kappa_{0}\cdot Len\}=d_{j}-\kappa_{0}\cdot Len$; then, it is expected that
   \begin{align*}
       r_{j}\cdot d_{j} & + \kappa_{0}\cdot (\delta_{j}-r_{j})\cdot Len\cdot \beta \\
       & + (d_{j}-\kappa_{0}\cdot Len)\cdot (\delta_{j}-r_{j}) \geq z_{j}.
   \end{align*}
   This leads to that $r_{j} \geq r_{j}^{\prime\prime}(\beta)$. As a summary of our analysis of both cases, the proposition holds. $\blacksquare$


\vspace{0.7em}\noindent\textbf{Proof of Proposition 4.2.}  When $x\in [0, \frac{d_{j}}{Len}-\kappa_{0})$, $g_{j}(x)=\max\{r_{j}^{\prime}(x), 0\}$; since $d_{j}\cdot\delta_{j}-z_{j}\geq 0$ and $(\kappa_{0}+1)\cdot Len>0$, $r_{j}^{\prime}(x)$ is a non-increasing function and so is $g_{j}(x)$. Similarly, when $x\in [\frac{d_{j}}{Len}-\kappa_{0}, 1)$, $g_{j}(x)=\max\{r_{j}^{\prime\prime}(x), 0\}$ is also non-increasing. In the rest of this proof, if suffices to show $g_{j}(x_{1})\geq g_{j}(x_{2})$ when $0\leq x_{1} < \frac{d_{j}}{Len}-\kappa_{0} \leq x_{2} < 1$. Given a job $j$, if $\kappa_{0}$ $=$ $0$, we have $g_{j}(x_{1})\geq 0 =g_{j}(x_{2})$. If $\kappa_{0}\geq 1$ and $d_{j}\cdot\delta_{j}=z_{j}$, we have $g_{j}(x_{1})=\delta_{j}=g_{j}(x_{2})$. If $\kappa_{0}\geq 1$ and $d_{j}\cdot\delta_{j}>z_{j}$, our analysis proceeds as follows. To prove $g_{j}(x_{1})\geq g_{j}(x_{2})$, it suffices to show $r_{j}^{\prime\prime}(x_{2})\leq r_{j}^{\prime}(x_{1})$; the function $r_{j}^{\prime\prime}(x)$ itself is non-increasing when $x\in [0, 1)$, and we have $r_{j}^{\prime\prime}(x_{2}) \leq  r_{j}^{\prime\prime}(x_{1})$. Hence, to prove $r_{j}^{\prime\prime}(x_{2})\leq r_{j}^{\prime}(x_{1})$, it suffices to prove $r_{j}^{\prime\prime}(x_{1})\leq r_{j}^{\prime}(x_{1})$, which can be proved by showing $A=(1-x_{1})\cdot\kappa_{0}\cdot Len \leq  d_{j}-(\kappa_{0}+1)\cdot Len\cdot x_{1}=B$. Since $x_{1}\in [0, \frac{d_{j}}{Len}-\kappa_{0})$, we have
\begin{center}
$B - A = d_{j}-(\kappa_{0}+x_{1})\cdot Len >0$.
\end{center}
Finally, the proposition holds. $\blacksquare$

\begin{figure*}[t]
	\centering
		\includegraphics[width=6.75in]{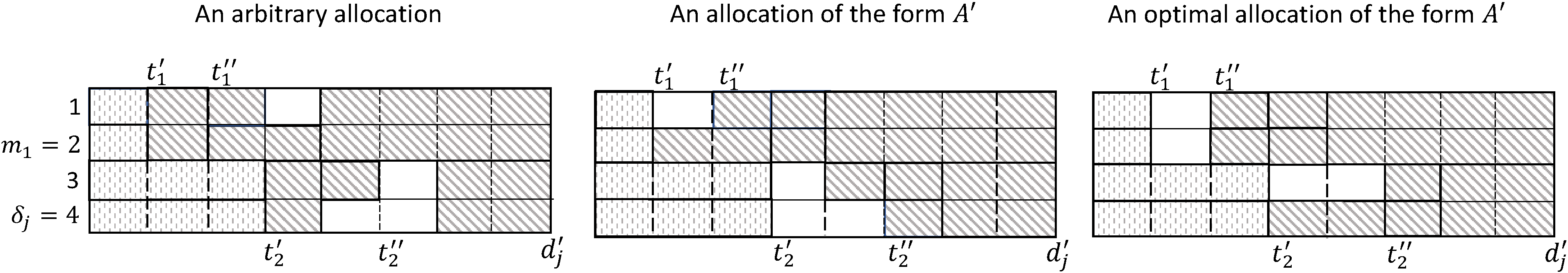}\\		
	\caption{Illustration for Proposition~4.7: the area of diagonal stripes denotes the allocation of on-demand instances to $j$.}\label{Fig.43}
\end{figure*}


\vspace{0.7em}\noindent\textbf{Proof of Proposition 4.4.} Firstly, we prove by contradiction that the optimal value of $o_{j}^{\kappa_{1}}$ is 0. Assume that $\hat{o}_{j}^{1}, \cdots, \hat{o}_{j}^{\kappa_{1}}$ are an optimal solution to (8) where $\hat{o}_{j}^{\kappa}\geq 1$. The constraint (6) has no effect on the value of $o_{j}^{\kappa_{1}}$. We can reduce the value of $\hat{o}_{j}^{\kappa_{1}}$ to 0; such reduction can still guarantee that (7) is satisfied, and $\hat{o}_{j}^{1}, \cdots, \hat{o}_{j}^{\kappa_{1}-1}, o_{j}^{\kappa_{1}}=0$ are a feasible solution to (8) under which (8) achieves a higher value, which contradicts that $\hat{o}_{j}^{1}, \cdots, \hat{o}_{j}^{\kappa_{1}}$ are an optimal solution to (8). Secondly, when $o_{j}^{\kappa_{1}}=0$, the objective function (8) equals $(\sum_{i=1}^{\kappa_{1}-1}{(\delta_{j}-o_{j}^{i})}+\delta_{j})\cdot Len\cdot\beta$. Under constraint (6), $\sum_{i=1}^{\kappa_{1}-1}{(\delta_{j}-o_{j}^{i})}\leq \frac{d_{j} \cdot \delta_{j} - z_{j}}{Len\cdot (1-\beta)}$. Since $o_{j}^{1}, \cdots, o_{j}^{\kappa_{1}-1}$ are integers, the maximum possible value of $\sum_{i=1}^{\kappa_{1}-1}{(\delta_{j}-o_{j}^{i})}$ is $\nu(z_{j}, d_{j})$. On the other hand, since $\delta_{j}-o_{j}^{i}\leq \delta_{j}$, the constraint (5) indicates that $\sum_{i=1}^{\kappa_{1}-1}{(\delta_{j}-o_{j}^{i})}\leq (\kappa_{0}-1)\cdot\delta_{j}$. Hence, the maximum possible value of $\sum_{i=1}^{\kappa_{1}-1}{(\delta_{j}-o_{j}^{i})}$ is $\min\{ \nu(z_{j}, d_{j}), (\kappa_{0}-1)\cdot\delta_{j}\}$. Now, we further show it is feasible. If $(\kappa_{0}-1)\cdot\delta_{j}\leq \nu(z_{j}, d_{j})$, $\sum_{i=1}^{\kappa_{1}-1}{(\delta_{j}-o_{j}^{i})}=(\kappa_{0}-1)\cdot\delta_{j}$ which leads to $\kappa_{0}-1\leq \kappa_{1}-1$; to satisfy (5), we have $\kappa_{1}=\kappa_{0}$. Then, constraint (7) holds trivially and constraint (6) is also satisfied. If $(\kappa_{0}-1)\cdot\delta_{j}>\nu(z_{j}, d_{j})$, $\sum_{i=1}^{\kappa_{1}-1}{(\delta_{j}-o_{j}^{i})}=\nu(z_{j}, d_{j})$; in this case, we have $\kappa_{1}-1\leq \kappa_{0}-1$. Furthermore, we also have $\nu(z_{j}, d_{j})+\delta_{j}>\frac{d_{j} \cdot \delta_{j} - z_{j}}{Len\cdot (1-\beta)}$ and (7) is satisfied. Finally, the proposition holds. $\blacksquare$


\vspace{0.7em}\noindent\textbf{Proof of Proposition 4.6.} We can check that when the strategy of utilizing spot instances is as above, $o_{j}^{1}, \cdots, o_{j}^{\kappa_{1}}$ are of the form in Proposition~4.4; hence, it is optimal. $\blacksquare$

\vspace{0.7em}\noindent\textbf{Proof of Proposition 4.7.} Let us consider an arbitrary allocation of on-demand instances to process the remaining $z_{j}^{i_{j}+1}$ workload, denoted by $\mathcal{A}$, also illustrated in Fig.~\ref{Fig.43} (left). These workload will be processed on $\delta_{j}$ instances, and let $x_{h}$ denote the total workload processed at the $h$-th instance where
\begin{equation}\label{constraint-on-demand-0}
 \sum\nolimits_{h=1}^{\delta_{j}}{x_{h}} \geq z_{j}^{i_{j}+1},
\end{equation}
\begin{equation}\label{constraint-on-demand-1}
\begin{split}
 x_{1}, \cdots, x_{m_{1}} \in [0, d_{j}^{\prime}-t_{1}^{\prime}+1]&,\\
 x_{m_{1}+1}, \cdots, x_{\delta_{j}} \in [0, d_{j}^{\prime}-t_{2}^{\prime}+1].
\end{split}
\end{equation}
The allocation $\mathcal{A}$ can be transformed into an allocation $\mathcal{A}^{\prime}$ with the following form without increasing the total cost of utilizing instances: the $x_{h}$ workload of the $h$-th instance is processed from the deadline $d_{j}^{\prime}$ towards earlier slots, i.e., in $[d_{j}^{\prime}-x_{h}+1, d_{j}^{\prime}]$,
which is illustrated in Fig.~\ref{Fig.43} (middle). Hence, in the following, we only need to show the cost-optimal strategy of utilizing instances when the allocation is of the form $\mathcal{A}^{\prime}$.

As illustrated in the Fig.~\ref{Fig.422}, let $\hat{\mathcal{I}}_{1}=[t_{1}^{\prime}, d_{j}^{\prime}]$ and $\hat{\mathcal{I}}_{2}=[t_{2}^{\prime}, d_{j}^{\prime}]$. From $d_{j}^{\prime}$ towards earlier slots in $\hat{\mathcal{I}}_{1}$ (resp. in $\hat{\mathcal{I}}_{2}$), let every $Len$ slots constitute a time interval, i.e., $\mathcal{I}_{i}=[d_{j}^{\prime}+1-i\cdot Len, d_{j}^{\prime}-(i-1)\cdot Len]$; for $\hat{\mathcal{I}}_{1}$ the last interval is $\mathcal{I}_{\hat{\kappa}_{1}+1} = [t_{1}^{\prime}, t_{1}^{\prime\prime}-1]$ (resp. for $\hat{\mathcal{I}}_{2}$  the last is $\mathcal{I}_{\hat{\kappa}_{2}+1} = [t_{2}^{\prime}, t_{2}^{\prime\prime}-1]$). Now, we describe the cost structure when the allocation of $j$ is of the form $\mathcal{A}^{\prime}$. We use $x_{h,i}$ to denote the workload processed by the $h$-th instance in $\mathcal{I}_{i}$ where for all $h\in [1, m_{1}]$,
\begin{equation}\label{constraint-on-demand-00-1}
 x_{h,1}, \cdots, x_{h,\hat{\kappa}_{1}} \in [0, Len],\enskip  x_{h,\hat{\kappa}_{1}+1}\in [0, \phi_{1}],
\end{equation}
and for all $h\in [m_{1}+1, \delta_{j}]$,
\begin{equation}\label{constraint-on-demand-00-2}
 x_{h,1}, \cdots, x_{h,\hat{\kappa}_{2}} \in [0, Len],\enskip  x_{h,\hat{\kappa}_{2}+1}\in [0, \phi_{2}].
\end{equation}
Let $\psi_{h}=\left\lceil \frac{x_{h}}{Len} \right\rceil$; under the allocation form of $\mathcal{A}^{\prime}$, we have for all $h\in [1, \delta_{j}]$ that
\begin{equation}\label{constraint-on-demand-2}
\begin{split}
& x_{h,1}= \cdots = x_{h,\psi_{h}-1}=Len,\\
& x_{h, \psi_{h}} = x_{h}-(\psi_{h}-1)\cdot Len,\\
& \text{the other } x_{h,i}=0,
\end{split}
\end{equation}
and
\begin{equation}\label{constraint-on-demand-equation}
\begin{split}
& x_{h} = \sum\nolimits_{i=1}^{\hat{\kappa}_{1}+1}{x_{h,i}},  \text{ if } h\in [1, m_{1}]\\
& x_{h} = \sum\nolimits_{i=1}^{\hat{\kappa}_{2}+1}{x_{h,i}},  \text{ if } h\in [m_{1}+1, \delta_{j}]
\end{split}
\end{equation}
where $0\leq x_{h, \psi_{h}}< Len$. We define the sign function $sgn(x)$: it equals 1 if $x>0$ and 0 if $x=0$. Let
\begin{equation}\label{constraint-on-demand-3}
y_{h,i}=sgn(x_{h,i})\in \{0, 1\},
\end{equation}
and the price of utilizing the $h$-th instance is $p$ times the sum of all $y_{h,i}$; here, by (\ref{constraint-on-demand-2}), the sum of all $y_{h,i}$ is $\psi_{h}$.

The cost minimization problem under the allocation form of $\mathcal{A}^{\prime}$ is as follows, referred to as $\boldsymbol{\mathcal{Q}}$\textbf{-\Rmnum{1}}:
\begin{align} \label{optimal-on-demand}
\text{min}\enskip \sum\limits_{h=1}^{m_{1}}{\sum\limits_{i=1}^{\hat{\kappa}_{1}+1}{p\cdot y_{h,i}}} + \sum\limits_{h=m_{1}+1}^{\delta_{j}}{\sum\limits_{i=1}^{\hat{\kappa}_{2}+1}{p\cdot y_{h,i}}}
\end{align}
subject to the constraints (\ref{constraint-on-demand-0})-(\ref{constraint-on-demand-3}). $\mathcal{Q}$-\Rmnum{1} corresponds to another optimization problem: its objective function is also (\ref{optimal-on-demand}), subject to (\ref{constraint-on-demand-0}), (\ref{constraint-on-demand-1}), (\ref{constraint-on-demand-equation}), (\ref{constraint-on-demand-3}), and for all $h\in [1, m_{1}]$
\begin{align}
 x_{h,1}, \cdots, x_{h,\hat{\kappa}_{1}}\in \{0, Len\},\enskip x_{h, \hat{\kappa}_{1}+1} \in \{ 0, \phi_{1} \},\label{constraint-on-demand-5-0}
\end{align}
and for all $h\in [m_{1}+1, \delta_{j}]$,
\begin{align}
 x_{h,1}, \cdots, x_{h,\hat{\kappa}_{2}}\in \{0, Len\},\enskip x_{h, \hat{\kappa}_{2}+1} \in \{ 0, \phi_{2} \}.\label{constraint-on-demand-5-1}
\end{align}
The above mathematical problem is referred to as $\boldsymbol{\mathcal{Q}}$\textbf{-\Rmnum{2}}. In the following, we prove that (\rmnum{1}) any solution to $\mathcal{Q}$-\Rmnum{1} corresponds to a solution to $\mathcal{Q}$-\Rmnum{2} and their objective function (\ref{optimal-on-demand}) under these two solutions achieves the same value; then, (\rmnum{2}) an optimal solution to $\mathcal{Q}$-\Rmnum{2} corresponds to a solution to $\mathcal{Q}$-\Rmnum{1}, and their objective function under these two solutions also achieves the same value. The first point shows that the optimal value of $\mathcal{Q}$-\Rmnum{2} is a lower bound of the optimal value of $\mathcal{Q}$-\Rmnum{1}. The second point shows that there is a solution to $\mathcal{Q}$-\Rmnum{1} under which the value of (\ref{optimal-on-demand}) equals the optimal value of $\mathcal{Q}$-\Rmnum{2}; hence, this solution to $\mathcal{Q}$-\Rmnum{1} is optimal and we will give such an optimal solution while proving the two points above.

The decision variables of both $\mathcal{Q}$-\Rmnum{1} and $\mathcal{Q}$-\Rmnum{2} are the same, i.e., $\{y_{h,i} | h\in [1,m_{1}], i\in [1, \hat{\kappa}_{1}+1]\}\cup \{y_{h,i} | h\in [m_{1}+1, \delta_{j}], i\in [1, \hat{\kappa}_{2}+1]\}$. Given a solution to $\mathcal{Q}$-\Rmnum{1} denoted by $Y$, we set the decision variables of $\mathcal{Q}$-\Rmnum{2} to the same values. Now, we show $Y$ is a feasible solution to $\mathcal{Q}$-\Rmnum{2}. Both in $\mathcal{Q}$-\Rmnum{2} and $\mathcal{Q}$-\Rmnum{1}, the same $x_{h,i}$ is set to non-zero and the others are set to zero by (\ref{constraint-on-demand-3}), and the non-zero's $x_{h,i}$ in $\mathcal{Q}$-\Rmnum{2} is $\geq$ the $x_{h,i}$ in $\mathcal{Q}$-\Rmnum{1} by (\ref{constraint-on-demand-00-1}), (\ref{constraint-on-demand-00-2}), (\ref{constraint-on-demand-5-0}), and (\ref{constraint-on-demand-5-1}). Since (\ref{constraint-on-demand-0}) holds in $\mathcal{Q}$-\Rmnum{1} where the value of $x_{h}$ is defined in (\ref{constraint-on-demand-equation}), we have (\ref{constraint-on-demand-0}) also holds in $\mathcal{Q}$-\Rmnum{2}. Hence, $Y$ is feasible. Furthermore, $\mathcal{Q}$-\Rmnum{1} and $\mathcal{Q}$-\Rmnum{2} have the same objective function (\ref{optimal-on-demand}) that achieves the same value under the same $Y$. This finishes proving the first point above.

Now, we give an optimal solution to $\mathcal{Q}$-\Rmnum{2}. The physical meaning of $\mathcal{Q}$-\Rmnum{2} can be explained as follows. There are 3 types of items each with a weight $p$: (\rmnum{1}) $\hat{\kappa}_{1}\cdot m_{1} + \hat{\kappa}_{2}\cdot m_{2}$ items each with a size $Len$, (\rmnum{2}) $m_{1}$ items each with a size $\phi_{1}$ ($< Len$), and (\rmnum{3}) $m_{2}$ items each with a size $\phi_{2}$ ($< Len$); the objective is to select some items such that the total size of chosen items is $\geq z_{j}^{i_{j}+1}$ (satisfying (\ref{constraint-on-demand-0})) while their total weight (i.e., (\ref{optimal-on-demand})) is minimized. Since items have the same weight, an optimal solution is just to select the minimum number of items, e.g., the items with the largest sizes, to exactly satisfy the size requirement; correspondingly, an optimal solution to $\mathcal{Q}$-\Rmnum{2} is such that the value of $y_{h,i}\in \{0, 1\}$ satisfies
\begin{equation}\label{constraint-on-demand-6}
\begin{split}
&y_{0}=\sum\limits_{h=1}^{m_{1}}{\sum\limits_{i=1}^{\hat{\kappa}_{1}}{y_{h,i}}}+\sum\limits_{h=m_{1}+1}^{\delta_{j}}{\sum\limits_{i=1}^{\hat{\kappa}_{2}}{y_{h,i}}},\\
&y_{1}=\sum\limits_{h=1}^{m_{1}}{y_{h,\hat{\kappa}_{1}+1}},\enskip y_{2}=\sum\limits_{h=m_{1}+1}^{\delta_{j}}{y_{h,\hat{\kappa}_{2}+1}},
\end{split}
\end{equation}
where $y_{0}, y_{1}, y_{2}$ are described in Proposition~4.7. We denote such a solution by $OPT_{2}$. Here, we set $x_{h,i}$ to non-zero if $y_{h,i}=1$ and zero otherwise by (\ref{constraint-on-demand-3}); the particular value of $x_{h,i}$ depends on (\ref{constraint-on-demand-5-0}) and (\ref{constraint-on-demand-5-1}), and it determines the value of $x_{h}$ by (\ref{constraint-on-demand-equation}) that can satisfy (\ref{constraint-on-demand-1}); by (\ref{constraint-on-demand-6}), $x_{h}$ can satisfy (\ref{constraint-on-demand-0}).

Next, we show $OPT_{2}$ corresponds to a solution $OPT_{1}$ to $\mathcal{Q}_{1}$-\Rmnum{1}, and their objective function~(\ref{optimal-on-demand}) under $OPT_{1}$ and $OPT_{2}$ achieves the same value.
In $\mathcal{Q}$-\Rmnum{1}, we set the value of $x_{h}$ to the same value when the solution to $\mathcal{Q}$-\Rmnum{2} is $OPT_{2}$ where the constraints (\ref{constraint-on-demand-0}) and (\ref{constraint-on-demand-1}) in $\mathcal{Q}$-\Rmnum{1} are naturally satisfied; then, we use (\ref{constraint-on-demand-2}) to obtain feasible $x_{h,i}$ that will also satisfy (\ref{constraint-on-demand-00-1}) and (\ref{constraint-on-demand-00-2}); by (\ref{constraint-on-demand-3}), the value of $y_{h,i}$ in $\mathcal{Q}$-\Rmnum{1} can be set, deriving a feasible solution $OPT_{1}$ to $\mathcal{Q}$-\Rmnum{1}. In both $\mathcal{Q}$-\Rmnum{1} and $\mathcal{Q}$-\Rmnum{2}, we have the number of non-zero's $y_{h,i}$ is $\left\lceil x_{h}/Len \right\rceil$; hence, $\mathcal{Q}$-\Rmnum{1} under $OPT_{1}$ and $\mathcal{Q}$-\Rmnum{2} under $OPT_{2}$ achieve the same value.
Finally, $OPT$ is an optimal solution to $\mathcal{Q}$-\Rmnum{1} by the two points above.

So far, we have completed the proof of Proposition 4.7. In the proof, we have given an optimal solution $OPT_{1}$ to the problem $\mathcal{Q}$-\Rmnum{1}; we can thus know the values of $y_{0}, y_{1}, y_{2}$ in (\ref{constraint-on-demand-6}) that correspond to a specific cost-optimal allocation of on-demand instances. This is also illustrated in Fig.~\ref{Fig.43} (right). $\blacksquare$

\section{The Online Learning Algorithm}

\begin{algorithm}[t]
	\SetKwInOut{Input}{Input}
	\SetKwInOut{Output}{Output}
	
	\Input{a set $\mathcal{P}$ of $n$ policies, each $\pi$ parameterized for indexing so that $\pi\in\{1, 2, \cdots, n\}$; the set $\mathcal{J}_{t}$ of jobs that arrive at $t$; }
	\BlankLine
	
    $i\leftarrow 1$\tcp*{\footnotesize{$i$ is used to track the number of times updating the weight distribution}}

	initialize the weight vector of policies: $w_{i}=\{w_{i, 1}, \cdots, w_{i, n}\}=\{1/n, \cdots, 1/n\}$\;

	\For{$t\leftarrow 1$ \KwTo $T$}{
         \tcp{\footnotesize{time goes from slot 1 towards later slots}}

            $\mathcal{J}_{t}^{\prime}\leftarrow \mathcal{J}_{t}$\;

            \While{$\mathcal{J}_{t}^{\prime}\neq \emptyset$}{

			    get a job $j$ from $\mathcal{J}_{t}^{\prime}$ such that $j$ is the smallest\;
			
			    pick a policy $\pi_{j}=\pi$ with a probability $w_{i,\pi}$, applied to $j$\tcp*{\footnotesize{When time $t$ goes from $a_{j}$ to $a_{j}+d_{j}-1$, the allocation of instances to $j$ is completed via the Algorithm~2}}

                $\mathcal{J}_{t}^{\prime}\leftarrow \mathcal{J}_{t}^{\prime}-\{j\}$\;

            }

	        	\If{$t > d$}{
		
                   $\mathcal{J}_{t-d}^{\prime\prime}\leftarrow \mathcal{J}_{t-d}$\;

			       \While{$\mathcal{J}_{t-d}^{\prime\prime}\neq \emptyset$}{
				
				        get a job $j^{\prime}$ from $\mathcal{J}_{t-d}^{\prime\prime}$ such that $j$ is the smallest\;

                        compute the cost of completing $j^{\prime}$ in the period of $[a_{j^{\prime}}, a_{j^{\prime}}+d_{j^{\prime}}-1]$ under every policy $\pi\in\mathcal{P}$, denoted by $c_{j^{\prime}}(\pi)$\tcp*{\footnotesize{When $t^{\prime}$ ranges from $a_{j^{\prime}}$ to $a_{j^{\prime}}+d_{j^{\prime}}-1$, the allocation to $j^{\prime}$ is completed via Dynalloc$\left(a_{j^{\prime}}, d_{j^{\prime}}, z_{j^{\prime}}^{\prime}, \delta_{j^{\prime}}, \beta_{0}, \beta, b, N_{\beta_{0}}, t^{\prime} \right)$; the cost is recorded accordingly}}

				       $\eta_{t}\leftarrow \sqrt{\frac{2\log{n}}{d(t-d)}}$\;
				
			        	\For{$\pi \leftarrow 1$ \KwTo $n$}{
		        			$w_{i+1,\pi}^{\prime}\leftarrow w_{i,\pi}\exp^{-\eta_{t}c_{j^{\prime}}(\pi)}$\;
				        }
				
        				\For{$\pi \leftarrow 1$ \KwTo $n$}{
		        			$w_{i+1,\pi}\leftarrow \frac{w_{i+1,\pi}^{\prime}}{\sum_{i=1}^{n}{w_{i+1, i}^{\prime}}}$\;
				        }

                       $i\leftarrow i+1$\;
				
				       $\mathcal{J}_{t-d}^{\prime\prime}\leftarrow \mathcal{J}_{t-d}^{\prime\prime}-\{j^{\prime}\}$\;

		         	}
	 	}
	}
	\caption{OptiLearning\label{Regret}}
\end{algorithm}

In this section, we formally describe the online learning algorithm (TOLA) used in Section~\ref{sec.learning} to learn the most cost-effective parameters  $\beta_{0}, \beta, b$. The online learning algorithm that we adopt is the one in \cite{Jain14}, presented as Algorithm~\ref{Regret}, and is also a form of the classic weighted majority algorithm. There are a set of jobs $\mathcal{J}$ that arrive sequentially over time, indexed by $j=1, 2, \cdots$, and a set of $n$ parametric policies $\mathcal{P}$ each specified by $\{\beta_{0}, \beta, b\}$ and indexed by $\pi=1, 2, \cdots$. Let $d=\max_{j\in\mathcal{J}}\{d_{j}\}$, i.e., the maximum relative deadline of all jobs. Let $\mathcal{J}_{t}\subseteq \mathcal{J}$ denote all jobs $j$ that arrive at time slot $t$, i.e., $a_{j}=t$. There is a weight distribution $w$ over $n$ policies; initially, it is a discrete uniform distribution $\{1/n, \cdots, 1/n\}$ (lines 1-2 of Algorithm~\ref{Regret}). The distribution $w$ will be updated as time goes by (lines 3, 9-20) and it is used to choose a policy in $\mathcal{P}$ for each job (lines 4-8).

Time $t$ goes from slot 1 to later slots (line 3). When a job $j\in\mathcal{J}_{t}$ arrives where $t=a_{j}$, the algorithm randomly picks a policy $\pi_{j}$ from $\mathcal{P}$ according to the current $w$ and bases the allocation of various instances to $j$ on that policy (lines 4-8). Let the policy $\pi_{j}$ be defined by $\{\beta_{0}^{(j)}, \beta^{(j)}, b^{(j)}\}$ and let the array $N$ denote the number of self-owned instances unreserved/available at every slot after allocating self-owned instances to the previous jobs $1, \cdots, j-1$ via the policy (4) with $\beta_{0}^{(1)}, \cdots, \beta_{0}^{(j-1)}$ as the control parameters respectively; initially, if $j=1$, we have $N(t)=R$ for all $t\in [1, T]$ where $R$ is the total number of self-owned instances. As time $t$ goes from slot $a_{j}$ towards $a_{j}+d_{j}-1$, the allocation of instances to $j$ is taken by executing Algorithm~2, i.e., Dynalloc$\left(a_{j}, d_{j}, z_{j}^{\prime}, \delta_{j}, \beta_{0}^{(j)}, \beta^{(j)}, b^{(j)}, N, t \right)$, at every slot $t\in [a_{j}, a_{j}+d_{j}-1]$ until $j$ is allocated enough instances to complete $z_{j}$ workload. As modeled in the Section~\ref{sec.model}, the cost of completing a job $j$ is from the use of spot and on-demand instances alone, and denoted by $c_{j}(\pi_{j})$.

On the other hand, when time goes to the beginning of slot $d+1$, the update of the weight distribution of policies begins (line 9). In particular, if $\mathcal{J}_{t-d} \neq \emptyset$, we sequentially consider every job $j^{\prime}$ in $\mathcal{J}_{t-d}$ (lines 10-12, 20). Let a virtual array $N_{\beta_{0}}$ denote the number of self-owned instances unreserved/available at every slot if the allocation of self-owned instances to the previous jobs $1, \cdots, j^{\prime}-1$ follows the policy (4) with the same control parameter $\beta_{0}$; initially, if $j^{\prime}=1$, we have $N_{\beta_{0}}(t)=R$ for all $t\in [1, T]$. For every policy $\pi\in\mathcal{P}$, it is defined by $\{\beta_{0}, \beta, b\}$; since the spot prices in $[t-d, t-1]$ have been revealed, we are able to compute the cost of completing the job $j^{\prime}$ in $[a_{j}, a_{j}+d_{j}-1]$ under the policy $\pi$ with $N_{\beta_{0}}$ recording the available self-owned instances, denoted by $c_{j^{\prime}}(\pi)$ (line 13). Subsequently, the weight of each policy (i.e., its probability) is updated so that the lower-cost (higher-cost) polices of this job are re-assigned the enlarged (resp. reduced) weights (lines 14-18).
Let $N^{\prime}=|\cup_{t=d+1}^{T}{\mathcal{J}_{t}}|$, i.e., the number of all jobs that arrive in $[d+1, T]$, and, as proved in \cite{Jain14}, we have that
\begin{proposition}\label{proposi-online-learning}
For all $\delta\in (0, 1)$, it holds with a probability at least $1-\delta$ over the random of online learning that
\begin{center}
$\max_{\pi\in\mathcal{P}}\left\{ \sum_{t\in\cup_{t=d+1}^{T}{\mathcal{J}_{t}}}{\frac{c_{j}(\pi_{j})-c_{j}(\pi)}{N^{\prime}}} \right\}\leq 9\sqrt{\frac{2d\log{(n/\delta)}}{N^{\prime}}}$.
\end{center}
\end{proposition}

Proposition~\ref{proposi-online-learning} says that, as TOLA runs, the actual total cost of completing all jobs is close to the cost of completing all jobs under a policy $\pi^{*}\in\mathcal{P}$ that generates the lowest total cost. Recall that a policy is defined by a tuple of parameters from $\mathcal{P}$.

\end{document}